\documentclass[12pt, onecolumn,a4paper,unpublished]{quantumarticle}
\pdfoutput=1 
\usepackage[T1]{fontenc}
\usepackage{makecell}
\usepackage{tikz}
\usetikzlibrary{positioning, calc}
\usetikzlibrary{arrows.meta}
\usetikzlibrary{arrows}
\usepackage{tikz-3dplot}
\usepackage{graphicx}
\usepackage{mdwlist}
\usepackage{color}
\usepackage{amsthm}
\usepackage{thmtools}
\usepackage{amssymb}
\usepackage{physics}
\usepackage{amsmath}
\usepackage{array}
\usepackage{doi}
\usepackage{url,hyperref}
\usepackage[capitalize,nameinlink]{cleveref}
\hypersetup{colorlinks={true},linkcolor={blue},citecolor=red}

\usepackage[numbers,sort&compress]{natbib}
\usetikzlibrary{fadings}
\usetikzlibrary{decorations.pathmorphing}
\usepackage{mathrsfs}
\usepackage{authblk}
\usepackage[mathscr]{euscript}
\usepackage{enumitem}

\usepackage{multicol}

\DeclareMathAlphabet\mathbfcal{OMS}{cmsy}{b}{n}

\usepackage{cleveref}

\usetikzlibrary{decorations.pathreplacing}
\tikzset{snake it/.style={decorate, decoration=snake}}
\tikzset{mytip/.tip={Triangle[length=2mm, width=2mm]}}
\tikzset{midbidir/.style={ thick, preaction={ decorate, decoration={ markings, mark=at position 0.3 with {\arrowreversed[]{mytip}}}}, postaction={decorate, decoration={markings, mark=at position 0.7 with {\arrow[]{mytip}} } }}}

\usetikzlibrary{decorations.pathreplacing,decorations.markings}

\tikzset{
  midarrow/.style={
    decoration={
      markings,
      mark=at position 0.5 with {\arrow{Triangle}}
    },
    postaction={decorate}
  }
}

\tikzset{
    >=stealth',
    punkt/.style={
           rectangle,
           rounded corners,
           draw=black, very thick,
           text width=6.5em,
           minimum height=2em,
           text centered},
    pil/.style={
           ->,
           thick,
           shorten <=2pt,
           shorten >=2pt,},
  on each segment/.style={
    decorate,
    decoration={
      show path construction,
      moveto code={},
      lineto code={
        \path [#1]
        (\tikzinputsegmentfirst) -- (\tikzinputsegmentlast);
      },
      curveto code={
        \path [#1] (\tikzinputsegmentfirst)
        .. controls
        (\tikzinputsegmentsupporta) and (\tikzinputsegmentsupportb)
        ..
        (\tikzinputsegmentlast);
      },
      closepath code={
        \path [#1]
        (\tikzinputsegmentfirst) -- (\tikzinputsegmentlast);
      },
    },
  },
  mid arrow/.style={postaction={decorate,decoration={
        markings,
        mark=at position .5 with {\arrow[#1]{stealth'}}
      }}}
}

\usepackage{hyperref}
\usepackage{slashed}
\usepackage{float} 
\usepackage{graphicx} 
\usepackage{subcaption}

\newtheorem{theorem}{Theorem}

\newtheorem{algorithm}[theorem]{Algorithm}

\newtheorem{corollary}[theorem]{Corollary}

\newtheorem{definition}[theorem]{Definition}

\newtheorem{lemma}[theorem]{Lemma}

\newtheorem{remark}[theorem]{Remark}

\newtheorem{protocol}[theorem]{Protocol}

\renewenvironment{proof}[1][Proof]{\noindent\textbf{#1. }}{\ \rule{0.5em}{0.5em}}
\tikzset{
    >=stealth',
    punkt/.style={
           rectangle,
           rounded corners,
           draw=black, very thick,
           text width=6.5em,
           minimum height=2em,
           text centered},
    pil/.style={
           ->,
           thick,
           shorten <=2pt,
           shorten >=2pt,},
  on each segment/.style={
    decorate,
    decoration={
      show path construction,
      moveto code={},
      lineto code={
        \path [#1]
        (\tikzinputsegmentfirst) -- (\tikzinputsegmentlast);
      },
      curveto code={
        \path [#1] (\tikzinputsegmentfirst)
        .. controls
        (\tikzinputsegmentsupporta) and (\tikzinputsegmentsupportb)
        ..
        (\tikzinputsegmentlast);
      },
      closepath code={
        \path [#1]
        (\tikzinputsegmentfirst) -- (\tikzinputsegmentlast);
      },
    },
  },
  mid arrow/.style={postaction={decorate,decoration={
        markings,
        mark=at position .5 with {\arrow[#1]{stealth'}}
      }}}
}

\begin{document}

\title{Entanglement summoning from entanglement sharing}

\author[1,3]{Lana Bozanic}
\email{lbozanic@uwaterloo.ca}
\orcid{}

\author[1,2]{Alex May}
\email{amay@perimeterinstitute.ca}
\orcid{0000-0002-4030-5410}

\author[1,2,4]{Stanley Miao}
\email{stanley.miao@uwaterloo.ca}
\orcid{0009-0000-7930-7563}

\affiliation[1]{Perimeter Institute for Theoretical Physics}
\affiliation[2]{Institute for Quantum Computing, Waterloo, Ontario}
\affiliation[3]{Department of Physics and Astronomy, University of Waterloo}
\affiliation[4]{Department of Combinatorics and Optimization, University of Waterloo}

\abstract{In an entanglement summoning task, a set of distributed, co-operating parties attempt to fulfill requests to prepare entanglement between distant locations.
The parties share limited communication resources: timing constraints may require the entangled state be prepared before some pairs of distant parties can communicate, and a restricted set of links in a quantum network may further constrain communication. 
Building on \cite{adlam2018relativistic, dolev2021distributing}, we continue the characterization of entanglement summoning. We give an if and only if condition on entanglement summoning tasks with only bidirected causal connections, and provide a set of sufficient conditions addressing the most general case containing both oriented and bidirected causal connections.
Our results rely on the recent development of entanglement sharing schemes \cite{Khanian2025}. 
}

\maketitle

\tableofcontents

\section{Introduction}

In this work, we are interested in how entanglement can be prepared and distributed among distant parties.
These parties are allowed to communicate using a limited set of communication links and are required to prepare entanglement in response to time-sensitive requests. 
This setting is formalized in a problem known as \emph{entanglement summoning}, which was introduced in \cite{adlam2018relativistic} and studied further in \cite{dolev2021distributing}. 

\begin{figure}
  \centering
  \begin{subfigure}{0.45\textwidth}
    \centering
    \begin{tikzpicture}[edge/.style={black, thick, midarrow}]

      \node[label=above left:$D_3$] (D3) {};
      \node[label=above right:$D_2$] (D2) at (3,0) {};
      \node[label=below:$D_1$] (D1) at (1.5,-2) {};

      \draw[fill=black] (D1) circle (0.1);
      \draw[fill=black] (D2) circle (0.1);
      \draw[fill=black] (D3) circle (0.1);
    
      \draw[edge] (D2) -- (D3);
      \draw[edge] (D1) -- (D2);
    
    \end{tikzpicture}
  \end{subfigure}
  \quad
  \begin{subfigure}{0.45\textwidth}
    \centering
    \begin{tikzpicture}[edge/.style={black, thick, midarrow}]

      \node[label=left:$D_1$] (D1) at (0,0) {};
      \node[label=right:$D_2$] (D2) at (3,0) {};
      \node[label=right:$D_4$] (D4) at (3,3) {};
      \node[label=left:$D_3$] (D3) at (0,3) {};

      \draw[fill=black] (D1) circle (0.1);
      \draw[fill=black] (D2) circle (0.1);
      \draw[fill=black] (D3) circle (0.1);
      \draw[fill=black] (D4) circle (0.1);
    
      \draw[midbidir] (D2) -- (D4);
      \draw[midbidir] (D1) -- (D2);
      \draw[midbidir] (D3) -- (D4);
    
    \end{tikzpicture}
  \end{subfigure}
  \caption{Graphs representing examples of entanglement summoning problems. The vertices $D_i$ represent network or spacetime locations where requests for entanglement may be received. Edges represent directed communication links. a) A simple example where all the communication links are singly-directed. b) An example where all of the communication links are bidirected.}\label{fig:simpleexample}
\end{figure}
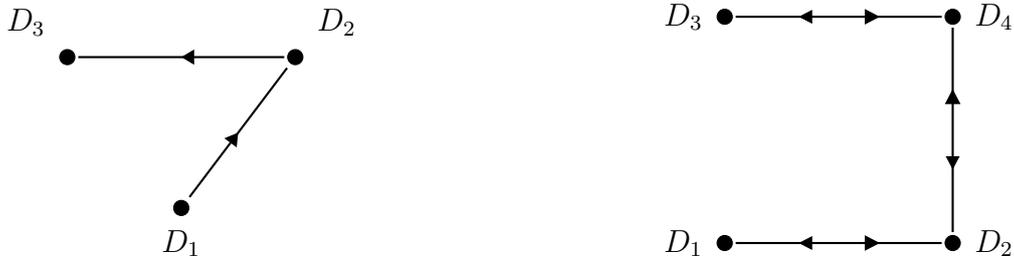

Examples of entanglement summoning problems are shown in figure \ref{fig:simpleexample}. 
An example is specified by a directed graph.
While we define the problem more abstractly in the main text, here the vertices can be taken to represent locations in a quantum network and the edges to represent allowed communication links. 
We consider a set of users who give requests at exactly two nodes in the network, and ask if a maximally entangled state can be prepared between those nodes and returned to the user. 
Note that we only allow a single round of communication, so edges $A\rightarrow B$ and $B\rightarrow C$ do not allow signals to be forwarded from $A$ to $C$. 
Our goal is to characterize when the causal constraints given by the network, combined with properties of quantum states, allow for this task to be completed. 

The entanglement summoning setting was initially motivated in the context of the quantum tasks framework \cite{kent2012quantum}, which asks for a general understanding of what quantum information processing tasks are possible when considering relativistic constraints. 
In this context, entanglement is needed between specified spacetime regions, rather than between spatial locations, to achieve certain goals. 
For instance, cheating strategies in quantum position verification \cite{Buhrman2014, Kent2011-qpv} require entanglement between appropriate spacetime regions. 
Entanglement summoning addresses when this can be achieved, and can be viewed as a basic building block for more elaborate spacetime quantum information processing scenario's. 

Entanglement summoning also builds on the study of single-system summoning \cite{kent2013no,hayden2016summoning}, which considers how a quantum system can be made available in response to time-sensitive requests in a network.
In that setting, summoning can be understood as characterizing operationally how quantum information can and cannot move through spacetime. 
We sometimes make use of single-system summoning in our work, and we take inspiration from its solution using secret-sharing schemes to address the bipartite, entangled case. 

The work \cite{dolev2021distributing} fully characterized entanglement summoning in the case where all pairs of locations in the network share only one-way causal connections: lab A can signal lab B, or B can signal A, but not both. 
Here, we fully characterize entanglement summoning in the case where all locations in the network share bidirected causal connections. 
In this case, the necessary and sufficient condition on the causal graph for entanglement summoning to be possible is simple, as we state next. 
\begin{theorem}\label{thm:bidirectedsummoning}
    In an entanglement summoning problem whose causal graph contains only bidirected edges, summoning is possible if and only if the cuasal graph admits a two-clique partition.  
\end{theorem}

The proof of theorem \ref{thm:bidirectedsummoning} makes use of a novel equivalence between performing entanglement summoning in a given graph and constructing a quantum state with corresponding entanglement properties. 
In particular, we show that entanglement summoning with a bipartite causal graph amounts to constructing an entanglement sharing scheme \cite{Khanian2025}. 

In an entanglement sharing scheme, non-communicating parties hold subsystems $S_1$,...,$S_n$. 
We consider subsets $T_i\subset \{S_1,...,S_n\}$, pairs of subsets $\{T_i, T_j\}$, and finally a set of these pairs, $\mathcal{A}=\{\{T_1, T_2\}, ...\}$. 
Given such a set $\mathcal{A}$, the goal is to construct quantum states $\ket{\Psi}_{S_1...S_n}$ such that the designated pairs of subsets of parties can recover a maximally entangled state by local operations. 
In \cite{Khanian2025}, a set of conditions on the list of designated pairs for it to be realized as an entanglement sharing scheme was understood. 
These are most easily stated in terms of an undirected graph $G_{\mathcal{A}}$, which has vertices labelled by the subsets $T_i$ and edges for each pairing $\{T_i,T_j\}\in \mathcal{A}$. 
We find that entanglement summoning with bidirected edges in a causal graph $G$ amounts to constructing an entanglement sharing scheme with sharing graph $G_{\mathcal{A}}$, given by taking the graph of causal connections defining the summoning problem, $G_C$, interpreting it as an undirected graph (by removing the arrows on all edges), and then taking the complement. 
An example of this transformation is shown in figure \ref{fig:pentagon-bidirected}. 

In the case of singly directed edges, we can also apply entanglement sharing schemes to give a protocol that addresses a broad set of cases. 
This class is large enough to include both the one-way only and bidirected only settings as special cases. 
We leave open whether this protocol addresses all feasible cases; in other words, we give sufficient conditions but leave open if they are necessary. 
The simplest uncharacterized entanglement summoning problem is shown in figure \ref{fig:uncharacterized}. 

\begin{figure}
    \centering
    \begin{tikzpicture}

    \foreach \i in {1,...,5}{
        \coordinate (V\i) at ({90+72*(\i-1)}:2cm);
      }
        \draw[line width=1pt, midbidir]
        (V1) -- (V2);
        \draw[line width=1pt, midbidir]
        (V2) -- (V3);
        \draw[line width=1pt, midbidir]
        (V3) -- (V4);
        \draw[line width=1pt, midbidir]
        (V4) -- (V5);
        \draw[line width=1pt, midbidir]
        (V5) -- (V1);
        \draw[line width=1pt, mid arrow]
        (V3) -- (V1);
        
      \foreach \i in {1,...,5}{
        \node at (V\i) {};
      }
      
    \node at ($(V1)+(0,0.3)$) {$D_1$};
    \node at ($(V5)+(0.3 ,0)$) {$D_2$};
    \node at ($(V4)+(0.3 ,0)$) {$D_3$};
    \node at ($(V3)+(-0.3 ,0)$) {$D_4$};
    \node at ($(V2)+(-0.3 ,0)$) {$D_5$};

    \end{tikzpicture}
    \caption{The simplest uncharacterized entanglement summoning problem. Removing the single-directed edge results in an impossible task, while making that edge bidirected results in a solvable task.}
    \label{fig:uncharacterized}
\end{figure}
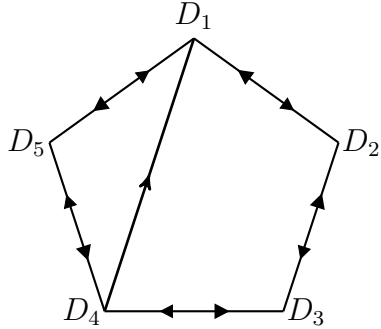

Our strategy for continuing the characterization of entanglement summoning problems was inspired by the solution to the single-system summoning problem \cite{hayden2016summoning}.
There, the first step in a solution is to observe that, given a causal graph, an error-correcting code which corrects a set of errors related to the structure of that graphs is necessary and sufficient to solve the summoning problem.  
Then the problem reduces to constructing an appropriate code.
Our approach proceeded by analogy to this one, with the modification that the appropriate object is no longer an error-correcting code, which stores quantum information, but instead the problem (in the bidirected case) reduces to the construction of an entanglement sharing scheme, which stores entanglement. 
Indeed, completing this analogy led us directly to the definition of entanglement sharing schemes and motivated the work \cite{Khanian2025}. 

Entanglement sharing and entanglement summoning together probe the limits and possibilities of entanglement structure in many-body quantum states. 
Entanglement sharing schemes address which patterns of entanglement can be realized in quantum states; entanglement summoning extends this to a dynamical setting. 

\subsection*{Acknowledgements}

AM, LB, and SM acknowledge the support of the Natural Sciences and Engineering Research Council of Canada (NSERC); this work was supported by an NSERC-UKRI Alliance grant (ALLRP 597823-24). We thank Adrian Kent for helpful discussions and comments.
Research at the Perimeter Institute is supported by the Government of Canada through the Department of Innovation, Science and Industry Canada and by the Province of Ontario through the Ministry of Colleges and Universities.

\section{Background}

\subsection{Entanglement summoning tasks}

The entanglement summoning scenario captures a wide variety of settings where entanglement needs to be distributed between distant parties. 
In particular, it explores how causal constraints allow or hinder entanglement preparation. 
We give a general definition of this scenario first, before describing a concrete setting where it arises. 
\begin{definition}\label{def:entanglement-summoning}
    An \textbf{entanglement summoning} task involves two players, Alice and Bob, and is defined by a directed graph ${G}_C$ known as the causal graph. 
    The graph specifies a set of communication channels that can be used in one simultaneous round to complete the task. 
    
    Before the task begins, Alice may distribute an arbitrary entangled state with shares held at each vertex. 
    To begin the task, at each vertex $i$, Bob provides Alice a single bit $b_i\in\{0,1\}$. 
    Alice is promised that exactly two of the input bits, $b_{j^*}$ and $b_{k^*}$, take the value $1$. 
    Alice can act on her locally-held systems at each vertex, use the communication channel defined by the graph, and then act locally again. 
    Alice succeeds in the task if she outputs the two subsystems of a fixed maximally entangled state, $\ket{\Psi^+}_{AB}$, at vertices $j^*$ and $k^*$.
\end{definition}

One setting where the entanglement summoning scenario arises is in a spacetime context, which we describe next. 
To capture the spacetime setting, we begin with a spacetime manifold $\mathcal{M}$. 
The manifold is a set of points and comes equipped with a metric, which allows us to define light cones and a notion of causality. 
In particular, if it is possible to travel from point $p$ to point $q$ without moving faster than the speed of light, we write $p\rightarrow q$. 
Given a point $p$, we define the sets
\begin{align}
    J^+(p) = \{q \in \mathcal{M}: p\rightarrow q\},\nonumber \\
    J^-(p) = \{q \in \mathcal{M}: q\rightarrow p\}.
\end{align}
$J^+(p)$ is referred to as the causal future of $p$; $J^-(p)$ is referred to as the causal past of $p$. 

To realize an entanglement summoning scenario in this spacetime setting, we suppose requests can be given at one set of spacetime points, and that shares from entangled states must be returned at another set of associated spacetime points. 
Concretely, consider a set of pairs of points $\{(c_i,r_i)\}_i$. 
At $c_i$, a \emph{call}, consisting of a bit $b_i\in\{0,1\}$, is given at $c_i$ to indicate whether or not a share of the returned maximally entangled state should be brought to the corresponding point $r_i$. 
Alice is guaranteed that exactly two of the calls, labelled $b_{j^*}$ and $b_{k^*}$, will take the value $1$, while the remaining calls will be $0$. 

Given $c_i$ and $r_i$, we define the casual diamond
\begin{align}
    D_i=J^+(c_i) \cap J^-(r_i).
\end{align}
This is the set of points that can see the call given at $c_i$, and influence the output at $r_i$. 
It turns out to be convenient to describe spacetime entanglement summoning scenarios in terms of their causal diamonds rather than as pairs of points $(c_i,r_i)$. 
We give the following definitions related to causal diamonds. 
\begin{definition}
    For two distinct causal diamonds, $D_i$ and $D_j$, we denote
    \begin{equation}
        D_i\to D_j
    \end{equation}
    \noindent if there exists a causal curve from a point in $D_i$ to a point in $D_j$. Moreover, we denote $D_i\stackrel{!}{\to} D_j$ to mean that $D_i\to D_j$ but $D_i\not\leftarrow D_j$. 
\end{definition}

\begin{definition}
    For two distinct causal diamonds $D_i$ and $D_j$ we denote
    \begin{equation}
        D_i\not\sim D_j    
    \end{equation}
    \noindent iff $D_i\not\to D_j$ and $D_i\not\leftarrow D_j$.
\end{definition}

From these notions of causal connections among diamonds, we can specify a causal graph. 
Given an entanglement summoning problem with diamonds $\{D_i\}_i$, we consider a graph ${G}$ with a vertex for each diamond. 
We then include a directed edge $(D_i,D_j)$ between vertices in ${G}$ whenever $D_i \rightarrow D_j$. 
Considering the set of strategies available in the spacetime setting, we can see that they are exactly captured by the notion of entanglement summoning given in definition \ref{def:entanglement-summoning}. 
In making this connection, we identify the graph describing the causal connections among diamonds $D_i$ as the causal graph ${G}_C$ appearing in definition \ref{def:entanglement-summoning}. 

Earlier work \cite{adlam2018relativistic,dolev2021distributing} identified the task of entanglement summoning generally with its spacetime realization, whereas we emphasize the somewhat abstract viewpoint of defining the problem in terms of a causal graph. 
Note that \cite{adlam2018relativistic, dolev2021distributing} define a number of variants of entanglement summoning. 
Our definition corresponds to unassisted entanglement summoning in their language.\footnote{It would also be interesting to attempt to characterize the assisted case using the strategies used here, but we leave this to future work. We believe doing so would require understanding variants of entanglement sharing schemes.}

Before continuing, we review some existing results on entanglement summoning. 
The first is a necessary condition on causal graphs for the task to be achievable. 
\begin{lemma}\label{lemma:no-two-out}
    \textbf{(No two-out)} Any entanglement summoning task corresponding to the ``two-out'' causal graph as in figure \ref{fig:no-two-out} or any spanning subgraph thereof is unachievable.
\end{lemma}
This was proven in \cite{adlam2018relativistic}.

\begin{figure}
    \centering
    \begin{tikzpicture}[edge/.style={black, thick, midarrow}]

      \node[label=above left:$D_3$] (D3) {};
      \node[label=above right:$D_2$] (D2) at (3,0) {};
      \node[label=below:$D_1$] (D1) at (1.5,-2) {};

      \draw[fill=black] (D1) circle (0.1);
      \draw[fill=black] (D2) circle (0.1);
      \draw[fill=black] (D3) circle (0.1);
    
      \draw[edge] (D2) -- (D3);
      \draw[edge] (D2) -- (D1);
    
    \end{tikzpicture}
    \caption{A causal graph corresponding to a ``two-out'' scenario.}
    \label{fig:no-two-out}
\end{figure}
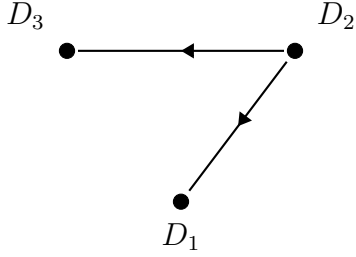

In \cite{dolev2021distributing}, the following partial characterization of entanglement summoning tasks was given. 
\begin{theorem}\label{thm:oriented-entanglement-summoning}
    \textbf{(Oriented entanglement summoning)} An entanglement summoning task whose causal graph $G_{C}$ is an oriented graph\footnote{By oriented we mean that for every pair of vertices, at most one of $(v_i,v_j)$ or $(v_j,v_i)$ is included.} is achievable if and only if for every $D_j\in\mathcal{D}$ the subgraph $G_{\mathcal{S}_j}$ induced by the subset
    \begin{equation}
        \mathcal{S}_j = \{D_i\in \mathcal{D} \,|\, D_i\not\to D_j\}
    \end{equation}
    \noindent is a tournament.\footnote{A tournament is an oriented graph where for every pair of vertices $i,j$, exactly one of $(v_i,v_j)$ or $(v_j,v_i)$ is included.}
\end{theorem}

\subsection{Entanglement sharing schemes}

Entanglement sharing was introduced in \cite{Khanian2025}. 
The task of entanglement sharing involves a dealer who prepares a state $|\Psi\rangle_{S_1...S_n}$. 
The dealer distributes the shares $S_1S_2...S_n$, with share $S_i$ given to party $i$. 
The goal is for specified sets of parties to be able to prepare a maximally entangled state by operations local to their subsystems, while other sets of parties should not be able to. 
To make this precise, we introduce the following definition. 
\begin{definition}
    A \textbf{pair access structure} $\mathcal{S}=(\mathcal{A}, \mathcal{U})$ on $n$ parties consists of a set of authorized pairs $\mathcal{A}=\{A_k\}_k$, and a set of unauthorized pairs $\mathcal{U}=\{U_\ell\}_\ell$.
    Each $A_k$ or $U_\ell$ is a pairing of two subsets, $\{T_i, T_j\}$ with $T_i$ and $T_j$ subsets of $\{S_i\}_{i=1}^n$. We require that any pairing $\{T_i,T_j\}$ occurring as an authorized or unauthorized set be disjoint, $T_i\cap T_j =\emptyset$.
\end{definition}
We require that the pairs $\{T_i,T_j\}\in \mathcal{A}$ can be used to recover a maximally entangled state, while pairs $\{T_i,T_j\}\in \mathcal{U}$ cannot. 
This is made more precise below. 

Entanglement sharing schemes come in two variants: given a pairing $\{T_i,T_j\}$, we can tell the party holding the set $T_i$ that they are trying to entangle with a partner holding $T_j$, or we can withhold this information. 
We refer to these cases as entanglement sharing with a known partner and entanglement sharing with an unknown partner, respectively. 
In this work, we will only make use of results on entanglement sharing schemes with an unknown partner. 
We define this more carefully next. 
\begin{definition}
    We say a state $\Psi_{S_1...S_n}$ is an \textbf{entanglement sharing scheme (ESS)} with an unknown partner and with access structure $\mathcal{S}$ if the following hold. 
    \begin{enumerate}
        \item There exists a family of channels $\{\mathcal{N}_{T_i\rightarrow X_i}\}$ such that for each $A_{k}=\{T_i,T_j\}\in \mathcal{A}$, 
        \begin{align}
            \mathcal{N}_{T_iT_j\rightarrow X_iX_j}(\Psi_{T_iT_j})= \Psi^+_{X_iX_j}
        \end{align}
        where we define $\mathcal{N}_{T_iT_j\rightarrow X_iX_j}=\mathcal{N}_{T_i\rightarrow X_i}\otimes \mathcal{N}_{T_j\rightarrow X_j}$.\footnote{By $\Psi^+_{AB}$ we mean the maximally entangled state on $AB$.}
        \item For each $U_k=\{T_i,T_j\}\in \mathcal{U}$, all possible channels $\mathcal{M}^{k}_{T_iT_j\rightarrow X_iX_j} = \mathcal{M}^k_{T_i\rightarrow X_i}\otimes \mathcal{M}^k_{T_j\rightarrow X_j}$ we have
        \begin{align}
            F(\mathcal{M}^{k}_{T_iT_j\rightarrow X_iX_j}(\Psi_{T_iT_j}),\Psi^+_{X_iX_j}) \leq 1-\delta.
        \end{align}
        for $\delta>0$.
    \end{enumerate} 
\end{definition}
We will use results from \cite{Khanian2025} that characterize when an entanglement sharing scheme can be achieved. 
The simplest way to understand this characterization is in terms of the following object. 
\begin{definition}
    The \textit{access-pair graph} associated to an entanglement sharing scheme $|\Psi\rangle_{S_1S_2...S_n}$ equipped with $\mathcal{C} = (\mathcal{A,U})$ is defined as follows. 
    Form the undirected graph $G_\mathcal{A}$ by letting each $T_i\subseteq \{S_i\}_{i=1}^n$ that occurs in $\mathcal{A}$ be a vertex, and each $\{T_i, T_j\}\in\mathcal{A}$ be an edge between vertices $T_i$ and $T_j$.
\end{definition}

We next state the following characterization of entanglement sharing schemes, which covers all cases where there are no unauthorized pairs, $\mathcal{U}=\emptyset$. 
\begin{theorem}\label{thm:ESSwithunknownpartner}
A pair access structure $\mathcal{S}=(\mathcal{A}, \varnothing)$ can be realized as an entanglement sharing scheme with an unknown partner if and only if the following condition holds. 
\begin{itemize}
    \item \textbf{Monogamy:} Suppose there is an even length path $\{T_1,T_2\}$, $\{T_2,T_3\}$, ...,$\{T_{k-1},T_k\}$ in $G_\mathcal{A}$. Then $T_1\cap T_k\neq \emptyset$. 
\end{itemize}
\end{theorem}
This theorem is proven in \cite{Khanian2025}, though we have stated this theorem somewhat differently than is done there. 
In \cite{Khanian2025}, they give a second condition that $G_\mathcal{A}$ have no odd cycles. 
This is actually implied by the monogamy condition as long as we impose $T_i\cap T_j=\emptyset$ whenever $\{T_i,T_j\}\in \mathcal{A}$, which we require for the pairing $\{T_i,T_j\}$ to be valid. 
We record this as the next lemma. 
\begin{lemma}\label{lemma:monogamygivesNOC}
    For valid sets of access pairs, the monogamy condition implies that there are no odd cycles in $G_\mathcal{A}$.
\end{lemma}
\begin{proof}
    We proceed by contradiction. 
    Suppose that $G_\mathcal{A}$ satisfies monogamy, but also contains an odd cycle, $\{T_{i_1}, T_{i_2}\},...,\{T_{i_{k-1}},T_{i_{k}}\},\{T_{i_{k}}, T_{i_1}\}$ for some $\{i_1, i_2,...,i_k\}\subseteq\{1,...,n\}$, where $k$ is odd. 
    Now begin with this cycle and remove $\{T_{i_k}, T_{i_1}\}$. 
    Since \{$T_{i_k},T_{i_1}$\} is an edge in $G_\mathcal{A}$, we have that $T_{i_k}\cap T_{i_1}=\emptyset$. 
    But this means we have found an even-length path for which the start and end points have no intersection, which violates monogamy, and is hence a contradiction. 
\end{proof}

Note that the no odd cycles condition does not imply the monogamy condition. 

\subsection{Graph theory results and terminology}

Throughout this work, we will use some language and results from graph theory.
We introduce some aspects of graph theory here, beginning with the definition of a graph. 
\begin{definition}
    A \textbf{directed graph} $G=(V,E)$ is a collection of vertices $V=\{v_1,...,v_n\}$, and collection of edges $E$ which are ordered pairings of vertices, e.g. $(v_i,v_j)$. 
    A pair of vertices $v_i$, $v_j$ are said to be connected if either $(v_i,v_j)\in E$ or $(v_j,v_i)\in E$. 
\end{definition}
Next we introduce the following terminology. 
\begin{definition}
    A \textbf{tournament} is a directed graph with exactly one edge between every pair of vertices. 
\end{definition}
\begin{definition}
    A \textbf{spanning subgraph} of a graph $G$ is a subgraph of $G$ that includes all of the vertices of $G$ (and a subset of the edges). 
\end{definition}
\begin{definition}
    A \textbf{path} in a directed graph is a collection of edges $(v_{i_1},v_{i_2}),(v_{i_2},v_{i_3}),..., \\ (v_{i_{k-1}},v_{i_k})$, where the second vertex of the $j$th edge is the same as the first vertex of the $j+1$st edge. 
\end{definition}
\begin{definition}
    A \textbf{cycle} in a directed graph is a path with the same starting and ending vertex.
\end{definition}

An undirected graph is analogous to a directed graph, but the edges are now unordered pairings, e.g. $\{v_i,v_j\}$. 
We next introduce a few notions relevant to undirected graphs.
\begin{definition}
    A \textbf{clique} of an undirected graph $G$ is a subset of vertices $S$ such that every pair of vertices in $S$ are connected. 
\end{definition}

\begin{definition}
    A \textbf{bipartite} graph is a graph $G$ which admits a decomposition of the vertices $V$ into $V_1$, $V_2$ such that $V=V_1\cup V_2$, $V_1\cap V_2=\emptyset$, and such that every edge in $G$ has one vertex in $V_1$ and one vertex in $V_2$. 
\end{definition}

\begin{definition}
    An undirected graph G=(V,E) is said to have a \textbf{two-clique partition} if its vertices can be divided into two disjoint sets, $K_1$ and $K_2$, such that $K_1\cup K_2=V$ each of which forms a clique.
\end{definition}

We will also need an adjacent notion of a clique in the case where we are concerned with directed graphs, which motivate the following definitions.

\begin{definition}
    Given a directed graph $G=(V,E)$, we will say that a subset $C\subseteq V$ forms a \textbf{quasi-clique} if for every pair of vertices $v_1,v_2\in C$, we have that at least one of $(v_1,v_2)$ or $(v_2,v_1)$ is in $E$.
\end{definition}

We will be interested in graphs that admit a two-quasi-clique partitioning. More specifically, we have the following definition.

\begin{definition}
    Given a directed graph $G=(V,E)$, we say that $G$ has a \textbf{two-quasi-clique partitioning} if there are two subsets, $K_1, K_2\subset V$ such that $K_1\cap K_2=\emptyset$, $K_1\cup K_2 = V$, and the induced subgraphs of $K_1$ and $K_2$ each form a quasi-clique.
    
\end{definition}

Next, we give a few basic results from graph theory. 
\begin{theorem}\label{thm:bipartite-is-noc}
    A graph is bipartite if and only if it does not contain any odd cycles. 
\end{theorem}
See \cite{diestel2025graphtheory} for a proof. 

\begin{theorem}\label{thm:cobipart-clique}
    A graph has a two-clique partition if and only if its complement is bipartite.
\end{theorem}
\begin{proof}
    If the complement of a graph $G = (V,E)$ is bipartite, $G^c = (V, E^c)$ can be divided into two sets of vertices, $V_1$ and $V_2$, such that $V_1 \cup V_2 = V$ and there are no edges connecting the vertices in each set by definition. This occurs if and only if the vertices in $V_1$ and $V_2$ are each separately completely connected in $G$, meaning $G$ can be divided into two cliques.
\end{proof}

We also present a corresponding theorem for directed graphs.

\begin{theorem}\label{thm:two-quasi-clique-partition}
    A directed graph has a two-quasi-clique partitioning if and only if its undirected complement\footnote{By the undirected complement of $G$, we mean take the complement of $G$ and then remove the directionality of the edges to obtain an undirected graph.} is bipartite.
\end{theorem}
\begin{proof}
    We first assume the graph has a two-quasi-clique partitioning and show its undirected complement is bipartite. 
    Denote by $K_1$ and $K_2$ the vertex sets of the partitioning. 
    Then in the undirected complement of $G$, the vertices in $K_1$ and $K_2$ are completely disconnected. 
    The complement of $G$ is hence bipartite.
    
    Next, suppose that the undirected complement of $G$ is bipartite. 
    Then there exists a partition of the vertices in ${G}^{\bar{c}}$,\footnote{We use the $\bar{c}$ notation for the undirected complement.} call it $\{K_1, K_2\}$, such that all edges have one end in $K_1$ and the other in $K_2$. 
    Then in the complement of $G^{\bar{c}}$, the vertices of $K_1$ become all-to-all connected, so that $K_1$ is a clique in the complement of $G^{\bar{c}}$. 
    Similarly, $K_2$ is a clique in the complement of $G^{\bar{c}}$. 
    Since the edges in the complement of $G^{\bar{c}}$ are undirected but correspond to directed edges in $G$, the induced subgraphs of $K_1$ and $K_2$ are in general quasi-cliques.
\end{proof}

We will refer to graphs with an bipartite undirected complement as \textbf{co-bipartite.} 

\section{Bidirected entanglement summoning}\label{sec:bidirected}

Recall that the case of entanglement summoning with oriented causal graphs was solved in \cite{dolev2021distributing}. 
In this section, we take up another special case of the general entanglement summoning problem: we restrict our attention to graphs where every edge is bidirected. 
In other words, the edge $(v_i,v_j)$ is included in the causal graph iff $(v_j,v_i)$ is. 
This set of examples turns out to have a simple solution. 
Note that an example of our technique was already given in \cite{Khanian2025}, where it was applied to the pentagon graph shown in figure \ref{fig:sidebyside}. 
We extend the same technique in this section to general bidirected graphs.

\subsection{Characterization in terms of the complement graph}

\begin{figure}
  \centering
  %
  \begin{subfigure}{0.45\textwidth}
    \centering
    \begin{tikzpicture}[scale=0.9]
      \foreach \i in {1,...,5}{
        \coordinate (V\i) at ({90+72*(\i-1)}:2cm);
      }
        \draw[line width=1pt, midbidir]
        (V1) -- (V2);
        \draw[line width=1pt, midbidir]
        (V2) -- (V3);
        \draw[line width=1pt, midbidir]
        (V3) -- (V4);
        \draw[line width=1pt, midbidir]
        (V4) -- (V5);
        \draw[line width=1pt, midbidir]
        (V5) -- (V1);
        
      \foreach \i in {1,...,5}{
        \node at (V\i) {};
      }
    \node at ($(V1)+(0,0.3)$) {$D_1$};
      \node at ($(V5)+(0.3 ,0)$) {$D_2$};
      \node at ($(V4)+(0.3 ,0)$) {$D_3$};
      \node at ($(V3)+(-0.3 ,0)$) {$D_4$};
      \node at ($(V2)+(-0.3 ,0)$) {$D_5$};

    \end{tikzpicture}
  \end{subfigure}
  \quad
  \begin{subfigure}{0.45\textwidth}
    \centering
    \begin{tikzpicture}[scale=0.9]
      \foreach \i in {1,...,5}{
        \coordinate (V\i) at ({90+72*(\i-1)}:2cm);
      }
      
      \foreach \i/\j in {1/3,2/4,3/5,4/1,5/2}{
        \draw[dashed, blue] (V\i) -- (V\j);
      }
      \foreach \i in {1,...,5}{
        \node at (V\i) {};
      }
      \node at ($(V1)+(0,0.3)$) {$T_1$};
      \node at ($(V5)+(0.3 ,0)$) {$T_2$};
      \node at ($(V4)+(0.3 ,0)$) {$T_3$};
      \node at ($(V3)+(-0.3 ,0)$) {$T_4$};
      \node at ($(V2)+(-0.3 ,0)$) {$T_5$};
    \end{tikzpicture}
  \end{subfigure}
  \caption{On the left is a bidirected entanglement summoning tasks involving $5$ diamonds, and on the right is the corresponding access-pair graph as described by the construction in protocol \ref{protocol:bidirected}.}\label{fig:pentagon-bidirected}
  
\end{figure}
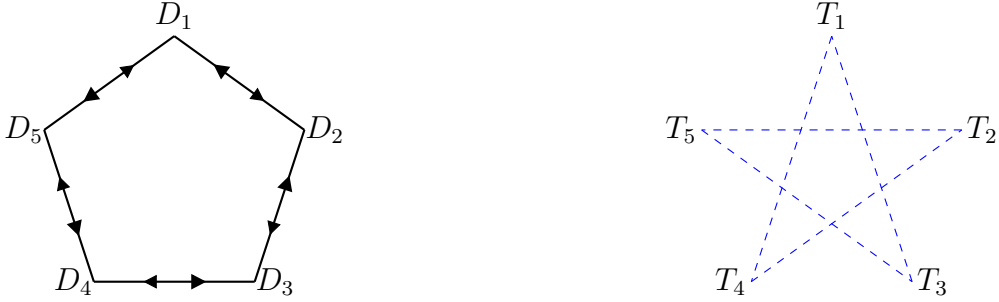

The first step in the characterization of the bidirectional case is to give a mapping between entanglement summoning and entanglement sharing.
This means we reduce from a question about the existence of a dynamical protocol to a question about the existence of a particular quantum state.
To do this, we describe a fully general protocol for completing summoning tasks with bidirected causal graphs. 

The first step in specifying a general protocol is to notice that we can ignore many of the possible inputs to the summoning task. 
In particular, we can distribute a dedicated maximally entangled state between every pair of connected diamonds. 
Then suppose there is a call to a pair of diamonds $D_i, D_j$ which are connected by an edge $\{D_i,D_j\}$.
The fact that this pair of calls has been made is apparent at both outputs $r_i$ and $r_j$, so we can discard any other systems used and return the dedicated maximally entangled state.
This means we can always respond correctly to neighbouring calls, so that the entanglement summoning task is possible if and only if it is possible when we restrict the inputs to non-neighbouring pairs of calls.

We next proceed to describe a protocol which addresses non-neighbouring calls. 
Define for each diamond $D_i$ the set of indices
\begin{align}
    K_i=\{j:D_j \leftrightarrow D_i\}.
\end{align}
Then the general protocol is as follows:
\begin{itemize}
    \item Pre-share a state $\ket{\Psi}_{X_1...X_n}$ with system $X_i$ brought to point $c_i$.  
    \item At $c_i$, after Alice receives call $b_i$, she applies $\mathcal{N}^{b_i}_{X_i\rightarrow Y[K_i]}$ where we define $Y^+[K_i]=\bigotimes_{j\in K_i}Y^{i\to j}$.
    \item Alice forwards system $Y^{i\to j}$ from $c_i$ to $r_j$. 
    \item At each $r_i$ with $b_i=1$, Alice applies a channel $\mathcal{M}_{Y^-[K_i]\rightarrow A_i}$ and returns $A_i$, where we defined $Y^-[K_i] = \bigotimes_{j\in K_i} Y^{j\to i}$. At the remaining $r_i$ she traces out all systems and gives no output. 
\end{itemize}
By showing this general protocol can be simplified, we will eventually identify a corresponding entanglement sharing scheme.
Note that $Y^+[K_i]=\bigotimes_{j\in K_i}Y^{i\to j}$ is the set of systems sent outward from vertex $i$ when $i$ does not receive a call, while the set $Y^-[K_i] = \bigotimes_{j\in K_i} Y^{j\to i}$ is the set of systems received at vertex $i$ when $i$ does receive a call.\footnote{By assumption, when $D_i$ receives a call none of its neighbours do, so they all forward shares to $D_i$.}
    
A first observation we can make towards simplifying this protocol is that, when Alice at $D_i$ receives $b_i=0$, we can have that Alice not keep any systems.
This is clear, as when $b_i=0$ Alice at vertex $i$ has no output.
Next, notice that since we only consider non-neighbouring calls, when Alice receives $b_i=1$ she can keep all her systems.
This is because in that case we have by assumption that there is no call to a neighbour, so they have no output, so there is no need to send systems to them.
This means we can take $\mathcal{N}^{b_i=0}=\mathcal{I}$ and have Alice keep $X_i$ when $b_i=1$. 
Overall then, the procedure at each diamond is to keep everything (when $b_i=1$) or send everything (when $b_i=0$). 

Next, consider that Alice can always replace her remaining non-trivial operation, $\mathcal{N}^{b_i}_{X_i\rightarrow Y[K_i]}$, with its isometric extension, which we call $V^i_{X_i\rightarrow Y[K_i]}$.\footnote{The purifying system can be absorbed into any one of the output systems.}
This means that rather than have Alice begin with the resource state $\ket{\Psi}_{X_1...X_n}$, we can instead give her $\ket{\Psi'}_{Y^+[K_1]...Y^+[K_n]}$, since she can always apply $(V^i_{X_i\rightarrow K_i})^\dagger$ and recover the earlier resource system $X_i$.

We record the simplified summoning protocol for bidirected edges below. 
\begin{protocol}\label{protocol:bidirected}\textbf{Bidirected edge summoning protocol:}
    The state $\ket{\Psi}_{Y^+[K_1]...Y^+[K_n]}$ is pre-distributed among the diamonds $D_i$, with system $Y^+[K_i]$ held at $D_i$. 
    Upon receiving the input $b_i$, the agent at vertex $i$ keeps all of $Y^+[K_i]$ if $b_i=1$, and send $Y^{i\to j}$ to each vertex $j\in K_i$ if $b_i=0$.
\end{protocol}

Finally, by using this simplified but still fully general protocol, we show entanglement summoning tasks in the setting of bidirected causal graphs require an associated ESS scheme be completed. 
\begin{lemma}\label{lemma:es-implies-essup}
    Consider an entanglement summoning task with a bidirectional causal graph $G_{C}$. 
    Consider the undirected complement of this causal graph, which we denote ${G}_C^c$. 
    Then the entanglement summoning task with causal graph $G_C$ can be completed if and only if the entanglement sharing scheme with graph $G_\mathcal{A}=G_C^c$ can be completed. 
\end{lemma}
\begin{proof}\,
    We claim that the state $\ket{\Psi}_{Y^+[K_1]...Y^+[K_n]}$ appearing in protocol \ref{protocol:bidirected} for entanglement summoning must be an entanglement sharing scheme for the access pair graph $G_C^c$. 
    
    To see why, consider the systems that the called-to locations $r_{j}$ and $r_{k}$ receive in this protocol. 
    At $j$, Alice keeps systems $Y^+[K_{j}]$ and receives system $Y^-[K_{j}]$, so that at vertex $j$ the systems $T_j=Y^+[K_{j}]\cup Y^-[K_{j}]$ are available. 
    Vertex $k$ similarly ends up with $T_k=Y^+[K_{k}]\cup Y^-[K_{k}]$.
    At this stage, a maximally entangled state must be prepared from these two collections of systems with local operations alone. 
    In other words, $\ket{\Psi}_{Y^+[K_1]...Y^+[K_n]}$ is exactly an entanglement sharing scheme, where $\{T_j,T_k\}$ must be authorized whenever $j,k$ are non-neighbouring vertices in the causal graph $G_C$, and hence are connected vertices in $G_C^c$. 
    Thus $G_\mathcal{A}=G_C^c$ as claimed. 
    
    Note that when $j$ and $k$ are non-neighbouring vertices, $T_j\cap T_k=\emptyset$, as required in the definition of an entanglement sharing scheme. 
    This occurs because a system starting at $j$ must be sent to $k$ or vice versa in some instance of the protocol for these systems to overlap, which cannot occur if they are non-neighbouring.
    Since the non-neighbouring pairs are exactly the edges in the complement of $G_C$,\footnote{Here we've interpreted $G_C$ as an undirected graph. To do this, each bidirected edge is replaced with an undirected edge.} the resulting entanglement sharing scheme has $G_\mathcal{A}=G_C^c$, as claimed.
\end{proof}

The above lemma reduces entanglement summoning to entanglement sharing, at least when the summoning problem involves only bidirected edges. 
Combined with theorem \ref{thm:ESSwithunknownpartner}, this tells us that entanglement summoning with bidirected edges is possible if and only if the complement of the causal graph, $G_C^c$, satisfies the monogamy condition.
Recall that the monogamy condition for entanglement sharing depends not only on the associated graph, but also on the structure of the associated subsets $T_i$. 
In our case, the subsets $T_i$ also happen to be fixed by the graph, so it should be possible to fully state our result as a condition on the graph $G_C$ alone. 
To do so, we begin with the following remark. 
\begin{remark}\label{remark:noedge}
    In the entanglement sharing scheme constructed from a summoning problem in lemma \ref{lemma:es-implies-essup}, the systems $T_j$ and $T_k$ overlap if and only if vertices $j$ and $k$ are neighbouring in $G_C$. That is, $T_j\cap T_k \neq \emptyset$ iff $D_i \sim D_k$.
\end{remark}
That these systems are non-overlapping when $j$ and $k$ are non-neighbouring was already noted in the proof of lemma \ref{lemma:es-implies-essup}. 
Conversely, if they are neighbouring, $T_k$ and $T_j$ both contain $Y^{j\to k}$ and $Y^{k\to j}$. 

The above lemma is sufficient data about the $T_i$ to reduce our condition on bipartite summoning to one stated purely in terms of the graph $G_C$.
In particular, we have the following. 
\begin{lemma}\label{lem:noodd-implies-monog}
    For an entanglement sharing scheme with an unknown partner whose access graph $G_{\mathcal{A}}$ generated from entanglement summoning, the following are equivalent.
    \begin{enumerate}
        \item $G_\mathcal{A}$ has no odd cycles.
        \item $G_\mathcal{A}$ and the associated authorized pairings $\{T_i,T_j\}$ satisfy monogamy.
    \end{enumerate}
\end{lemma}
\begin{proof}
From lemma \ref{lemma:monogamygivesNOC}, we had already that the monogamy condition implies no odd cycles.

We now show that $G_\mathcal{A}$ having no odd cycles implies that $G_\mathcal{A}$ satisfies monogamy.
Suppose that $G_\mathcal{A}$ has no odd cycles. 
This implies that for every even path $\{T_{i_1}, T_{i_2}\},...,\{T_{i_{k-1}}, T_{i_k}\}$ for some $\{i_1, i_2,...,i_k\}\subseteq\{1,...,n\}$, it is guaranteed that $\{T_{i_1}, T_{i_k}\}$ is \textbf{not} authorized, hence non-neighbouring in $G_C^c$, and so is neighbouring in $G_C$. 
By remark \ref{remark:noedge}, it follows that $T_{i_1} \cap T_{i_k} \neq \emptyset$. 
Hence, monogamy is satisfied.
\end{proof}

Finally, we are ready to collect our results so far and prove the following characterization of entanglement summoning tasks in the bidirected setting.
\begin{theorem}\label{thm:bidirectedandNOC}
    (\textbf{Bidirected entanglement summoning}) An entanglement summoning task whose causal graph $G_C$ is a bidirected graph is achievable if and only if the complement graph, $G^c_C$ contains no odd cycles.
\end{theorem}
\begin{proof}
    From lemma \ref{lemma:es-implies-essup}, we have that the entanglement summoning problem can be solved if and only if an entanglement sharing problem with access pair graph $G_C^c$ can be. 
    From theorem \ref{thm:ESSwithunknownpartner}, we have that the entanglement sharing scheme is possible if and only if $G_C^c$ satisfies monogamy. 
    From lemma \ref{lem:noodd-implies-monog}, we get that in this setting, the monogamy and no odd cycles conditions are equivalent, so we have that the summoning problem is possible if and only of $G_C^c$ contains no odd cycles, as needed. 
\end{proof}

\begin{figure}
  \centering
  %
  \begin{subfigure}[t]{0.45\textwidth}
    \centering
    \begin{tikzpicture}[scale=1, every node/.style={circle,fill=black,inner sep=1.5pt}]
    
        \node (T1) at (0,3) {};
        \node (T2) at (3,3) {};
        \node (T3) at (3,0) {};
        \node (T4) at (0,0) {};

        \draw[line width=1pt, midbidir]
        (T1) -- (T2);
        \draw[line width=1pt,midbidir]
        (T1) -- (T3);
        \draw[line width=1pt,midbidir]
        (T1) -- (T4);
      \draw[dashed, blue]
      (T2) -- (T3) -- (T4) -- (T2)-- cycle;

    \end{tikzpicture}
  \end{subfigure}
  \quad
  %
  \begin{subfigure}[t]{0.45\textwidth}
    \centering
    \begin{tikzpicture}[scale=1, every node/.style={circle,fill=black,inner sep=1.5pt}]
      \foreach \i in {1,...,6}{
        \coordinate (H\i) at ({60+60*(\i-1)}:2cm);
      }
        \draw[line width=1pt, midbidir]
        (H1) -- (H2);
        \draw[line width=1pt, midbidir]
        (H2) -- (H3);
        \draw[line width=1pt, midbidir]
        (H3) -- (H4);
        \draw[line width=1pt, midbidir]
        (H4) -- (H5);
        \draw[line width=1pt, midbidir]
        (H5) -- (H6);
        \draw[line width=1pt, midbidir]
        (H6) -- (H1);
        \draw[line width=1pt, midbidir]
        (H1) -- (H4);
      \draw[dashed, blue]
        (H1) -- (H3) -- (H5) -- cycle;
        \draw[dashed, blue]
        (H2) -- (H4) -- (H6) -- cycle;
      \foreach \i in {1,...,6}{
        \node at (H\i) {};
      }
    \end{tikzpicture}

    \label{fig:hexagon}
  \end{subfigure}
  \caption{Examples of unachievable bidirectional entanglement summoning tasks. Both causal graphs contain odd cycles in their complement.}
  \label{fig:sidebyside}
\end{figure}
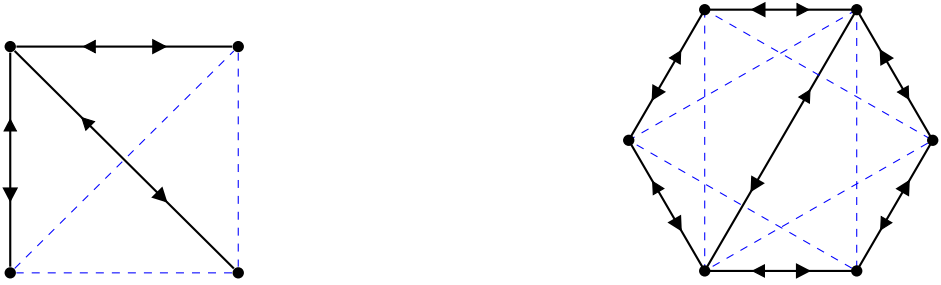

\subsection{Characterization in terms of the causal graph}\label{section:physical-bidirected}

The above analysis gives a characterization of achievable bidirected entanglement summoning tasks in terms of its complement graph, ${G}_C^c$. 
In particular, we found that there should be no odd cycles in the complement graph. 
To interpret this condition, it is useful to rephrase this as a condition directly on the original graph. 

To do this, we make use of theorem \ref{thm:cobipart-clique}, which allows us to restate theorem \ref{thm:bidirectedandNOC}
\begin{theorem}\label{thm:es-when-twoclique}
    A bidirected entanglement summoning task is achievable if and only if it admits a two-clique partition.
\end{theorem}
This result is already suggestive: it suggests that the two cliques occurring in the partition are somehow associated with the two subsystems of the maximally entangled state that should be returned. 

We can make this more precise by making use of a result on single-system summoning. 
In a single system summoning problem, there is a collection of diamonds $\{D_i\}$ which share causal connections described by a graph $G_C$. 
A single diamond will receive a call, and the goal is to return a quantum system $A$ there which is in an unknown (to Alice, the player) state. 
The following result from \cite{hayden2016summoning} characterizes when this is possible. 
\begin{theorem}\label{thm:single}
A single-system summoning task on diamonds $\{D_i\}$ is achievable if and only if the causal graph $G_C$ contains a tournament as a spanning subgraph. 
\end{theorem}

This theorem leads to the following interpretation of our result that entanglement summoning in bidirected graphs is possible iff the causal graph admits a two-clique partition. 
To complete the entanglement summoning task, we take a single maximally entangled state $\ket{\Psi^+}_{AB}$ and associate $A$ to the first subset of the partition $P_1$ and $B$ to the second subset of the partition $P_2$. 
We then execute a single system summoning protocol to move $A$ through $P_1$ and another protocol to move $B$ through $P_2$. 
Because a clique contains a tournament as a spanning subgraph,\footnote{Recall that in our clique the undirected edges $\{v_i,v_j\}$ correspond to having both directed edges $(v_i,v_j)$, $(v_j,v_i)$ in the (directed) causal graph.} this is possible according to theorem \ref{thm:single}. 
This successfully responds to calls to pairs of diamonds that are across the partition. 
Since every vertex within a partition is connected by a bidirected edge, calls to diamonds within the partitions can be handled by using dedicated maximally entangled states for each edge (as we also used in the last section).
Thus, our result amounts to the statement that this somewhat obvious protocol is actually the most general one. 

\section{General Summoning Tasks}\label{sect:general-summoning-tasks}

We now turn to more general summoning tasks. 
So far, entanglement summoning tasks whose causal graphs only have directed edges have been characterized by the work in \cite{dolev2021distributing} and the bidirected case by our analysis in section \ref{sec:bidirected}. 
The final task is to address the mixed case. 
We fall short of a complete characterization in the general case. 
However, we find that by extending our entanglement sharing strategy to the mixed case, we can give a protocol general enough to implement all feasible singly-directed and bidirected examples, and many mixed examples with both types of edges as well. 
We leave open if the examples we can complete in this way constitute all feasible cases. 

\subsection{Construction of the associated sharing scheme}

We begin by finding an associated entanglement sharing scheme starting from a summoning task with mixed causal edges.
The construction is similar in spirit to the bidirected case handled in the last section, but it becomes somewhat more involved. 
More concretely, using the causal graph $G_C$, we will construct an associated access pair graph $G_{\mathcal{A}}$. 
THe construction is such that an entanglement sharing scheme for $G_{\mathcal{A}}$ leads to a protocol for the summoning problem with causal graph $G_C$. 
We begin by defining the following objects. 
\begin{definition}\label{def:inoutbisets}
    For each causal diamond $D_i$ in a causal graph with $n$ total diamonds, we define the following index sets.
    \begin{enumerate}
        \item The \textbf{out edge set} $\mathcal{O}_i$ is defined to be the subset of indices $k\in \{1,..,n\}$ such that $D_i\stackrel{!}{\to} D_k$. If a causal diamond has no outgoing edges, we define $\mathcal{O}_i=\{i\}$.
        \item The \textbf{in edge set} $\mathcal{I}_i$ is defined to be the subset of indices $k\in\{1,...,n\}$ such that $D_k\stackrel{!}{\to} D_i$.
        \item The \textbf{bidirected edge set} $\mathcal{B}_i$ is the set of indices $k\in \{1,...,n\}$ such that $D_i \leftrightarrow D_k$.
    \end{enumerate}
\end{definition}
We have that $\mathcal{O}_i \cap \mathcal{I}_i = \emptyset$ for all $i\in\{1,..n\}$.

Note that any calls to neighbouring diamonds which are connected by a bidirected edge can be handled by a dedicated maximally entangled state, as was also done in the last section, so we can ignore those cases. 
Next, given a set of causal diamonds $D_i$, we define a state $\ket{\Psi}_{Y^+[\mathcal{O}_1]...Y^+[\mathcal{O}_n]}$ with $Y^+[\mathcal{O}_i]$ held at $D_i$, and where we define
\begin{align}
    Y^+[\mathcal{O}_i] &= \bigcup_{j\in \mathcal{O}_i\cup\mathcal{B}_i} Y^{i\to j},\nonumber \\
    Y^-[\mathcal{I}_i] &= \bigcup_{j \in \mathcal{I}_i\cup\mathcal{B}_i} Y^{j\to i}.
\end{align}
Then we consider the following protocol. 
\begin{protocol}\label{protocol:mixeddirected}\textbf{Entanglement summoning protocol:}
The state $\ket{\Psi}_{Y^+[\mathcal{O}_1]...Y^+[\mathcal{O}_n]}$ is pre-distributed among the diamonds $D_i$, with system $Y^+[\mathcal{O}_i]$ held at $D_i$. 
Upon receiving the input $b_i$, the agent at vertex $i$ keeps all of $Y^+[\mathcal{O}_i]$ if $b_i=1$, and sends $Y^{i\to j}$ to each vertex $j\in \mathcal{O}_i\cup \mathcal{B}_i$ if $b_i=0$.
\end{protocol}
We claim this protocol defines an entanglement sharing scheme. 
To see how, consider the subsets of systems that become available at each vertex given different possibilities for the inputs $b_i$. 
If $D_i$ receives a call and none of its neighbours do, then at $r_i$ the systems
\begin{equation}\label{eq:Tidef}
    T_i = \left[\bigcup_{k\in\mathcal{I}_i} Y^{k\to i}\right]\cup \left[\bigcup_{l\in\mathcal{O}_i} Y^{i\to l}\right]
\end{equation}
become available. 
We add a vertex to the graph of the sharing scheme for each $T_i$. 
Notice that $T_i$ can never be empty. 
This is assured because either $D_i$ has an edge pointing outward to another diamond $D_j$, in which case $Y^{i\rightarrow j}\in T_i$, or $D_i$ has no outgoing edges, in which case we took $\mathcal{O}_i=\{i\}$ so that there is a system $Y^{i\rightarrow i}\in T_i$. 

Pairs of such systems $\{T_i,T_j\}$ with $D_i \not\sim D_j$ must be authorized pairs: since the two diamonds are non-neighbouring, they each receive the subsets $T_i$ and $T_j$ respectively with $T_i\cap T_j=\emptyset$\footnote{Note that remark \ref{remark:noedge} continues to hold in the setting of this section.}, and hence it must be possible to recover a maximally entangled state from these two subsystems. 
To represent this, we add an edge to the sharing access pair graph for each such $\{T_i,T_j\}$. 

Next, we move on to consider calls to neighbouring diamonds.
Recall that we could separately handle the case where the diamonds share a bidirected edge and both receive a call, so we are interested in the case where they share a singly directed edge.
Suppose we have two diamonds $D_i,D_j$ with $D_j\rightarrow D_i$.
Define the set $T_{i\backslash j}$ to be the systems diamond $i$ receives on inputs $b_i=b_j=1$. 
Considering the above protocol, this set will be
\begin{align}\label{eq:Tijdef}
    T_{i\backslash j} = \left[\bigcup_{l\in\mathcal{I}_i\backslash\{j\}} Y^{l\to i}\right]\cup \left[\bigcup_{k\in\mathcal{O}_i} Y^{i\to k} \right]  = T_i\setminus \{Y^{j\to i}\}.
\end{align}
In words, $T_{i\backslash j}$ consists of all those systems given to $i$ initially (which are kept because $b_i=1$) plus almost all of those systems associated with edges that point in to $D_i$, with the exception being $Y^{j\to i}$ (because $b_j=1$ so all systems starting at $D_j$ are kept there). 
Notice that $T_{i\backslash j}$ can never be empty. 
Again, this is because there is always at least one system that starts at $D_i$ (because $\mathcal{O}_i$ is defined such that it is never the empty set), and having the call to $j$ can only prevent systems from arriving at $D_i$, but it cannot remove the systems that start there. 

Considering again calls to $b_i=b_j=1$, notice that the diamond $D_j$ does not see the call at $D_i$, so it receives the same systems $T_j$ as before. 
This means $\{T_i, T_{j\backslash i}\}$ must be an authorized pair in the sharing scheme so that the protocol can successfully respond to this pair of calls.

We can summarize the construction of the access pair graph as follows.
\begin{algorithm}\label{algorithm:accesspairgraphconstruction}
    Given a summoning problem defined by a causal graph $G_C$, we construct an access pair graph $G_{C\rightarrow \mathcal{A}}$ according to the following prescription. 
    Note that $G_C$ is a directed graph while $G_{C\rightarrow \mathcal{A}}$ is undirected.
    \begin{enumerate}
        \item For each vertex $i\in G_C$, add a vertex labelled $T_i$ to $G_{C\rightarrow \mathcal{A}}$. Vertices added in this way are known as \textbf{body} vertices.
        \item For every pair of vertices $i,j$ with no edge between them in $G_C$, add an edge $\{T_i,T_j\}$ in $G_{C\rightarrow \mathcal{A}}$. 
        \item For every pair of vertices $D_i,D_j$ in the causal graph with $D_i\stackrel{!}{\rightarrow} D_j$, add a vertex labelled $T_{j\setminus i}$ to $G_{C\rightarrow \mathcal{A}}$. Vertices added in this way are known as \textbf{wing} vertices. Further, add an edge $\{T_i,T_{j\backslash i}\}$ to the edge set of $G_{C\rightarrow \mathcal{A}}$.  
    \end{enumerate}
\end{algorithm}
Notice that the symbols $T_i$, $T_{j\backslash i}$ are used for two meanings: they label vertices in $G_{C\rightarrow \mathcal{A}}$, and they label subsets of quantum systems. 
Note that the construction never creates edges between wing vertices. 
Also notice that the subgraph defined by the body vertices is exactly $G_C^c$. 

\begin{figure}[h]
    \centering
    \begin{subfigure}[t]{0.45\textwidth}
    \centering
    \begin{tikzpicture}[
    >=Triangle,  
    node/.style={circle, fill=black, inner sep=1.5pt},
    lab/.style={font=\small, inner sep=2pt},
    edge/.style={black, thick, midarrow} 
  ]

  \node[node] (T1) at (0,2) {};
  \node[lab, above=2pt of T1] {$D_1$};

  \node[node] (T2) at (2,2) {};
  \node[lab, above=2pt of T2] {$D_2$};

  \node[node] (T3) at (2,0) {};
  \node[lab, below=2pt of T3] {$D_3$};

  \node[node] (T4) at (0,0) {};
  \node[lab, below=2pt of T4] {$D_4$};

  \draw[edge] (T2) -- (T3);
  \draw[edge] (T2) -- (T1);
  \draw[edge] (T3) -- (T4);

\end{tikzpicture}
\subcaption{Causal graph.}
\label{subfig:causal-example}
    \end{subfigure}
    \begin{subfigure}[t]{0.45\textwidth}
\begin{tikzpicture}[>=stealth, 
    node/.style={circle, fill=black, inner sep=1.5pt},
    lab/.style={font=\small, inner sep=2pt},
    green edge/.style={green!70!black, thick},
    blue edge/.style={blue!70!black, thick}
  ]

  \node[node] (T1) at (0,2) {};
  \node[lab, above=2pt of T1] {$T_1$};

  \node[node] (T2) at (2,2) {};
  \node[lab, above=2pt of T2, xshift=-6pt] {$T_2$};

  \node[node] (T3) at (2,0) {};
  \node[lab, below=2pt of T3] {$T_3$};

  \node[node] (T4) at (0,0) {};
  \node[lab, below=2pt of T4] {$T_4$};

  \draw[green edge] (T1) -- (T4);
  \draw[green edge] (T1) -- (T3);
  \draw[green edge] (T2) -- (T4);

  \node[node] (T1s) at ($(T2)+(0,1.5)$) {};
  \node[lab, above=2pt of T1s] {$T_{1\backslash 2}$};

  \node[node] (T3s) at ($(T2)+(1.5,0)$) {};
  \node[lab, right=2pt of T3s] {$T_{3\backslash 2}$};

  \draw[blue edge] (T2) -- (T1s);
  \draw[blue edge] (T2) -- (T3s);

  \node[node] (T4s) at ($(T3)+(1.5,0)$) {};
  \node[lab, right=2pt of T4s] {$T_{4\backslash 3}$};
  \draw[blue edge] (T3) -- (T4s);

    \end{tikzpicture}
    
    \subcaption{Corresponding access-pair graph}
    \label{subfig:access-pair-example}
    \end{subfigure}
    \caption{An entanglement summoning task whose causal graph has directed edges, and its corresponding access-pair graph.}
    \label{fig:directed-sq}
\end{figure}
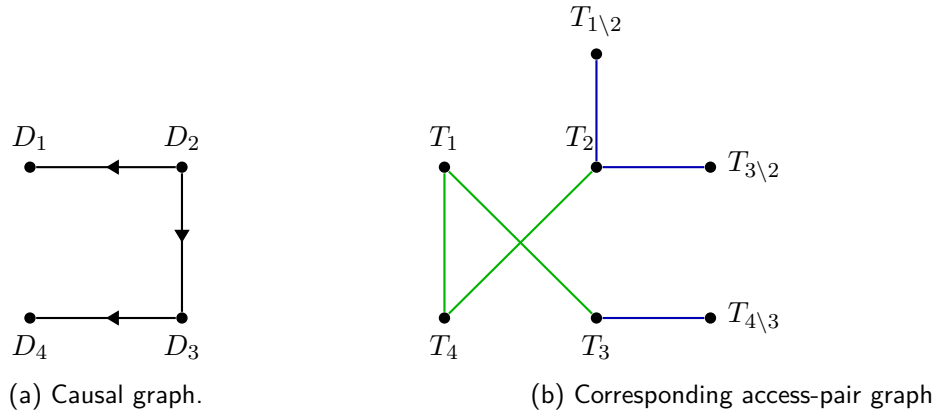

An example of this construction is depicted in figure \ref{fig:directed-sq}. 
Figure \ref{subfig:causal-example} depicts a summoning task where  $D_2\to D_3, D_1$ and $D_3\to D_4$. The corresponding access-pair graph is the complement, with additional edges at $T_2$ to $T_{1 \backslash 2}$ and $T_{3 \backslash 2}$, and at $T_3$ to $T_{4 \backslash 3}$. 
The subsystems associated to each vertex are:
\begin{multicols}{2}
\begin{itemize}
    \item $T_1 = Y^{1\to 1} Y^{2\to 1}$
    \item $T_2 = Y^{2\to 1} Y^{2\to 3}$
    \item $T_3 = Y^{3\to 4} Y^{2\to 3}$
    \item $T_4 = Y^{4\to 4} Y^{3\to 4}$
\end{itemize}
\columnbreak
\begin{itemize}
    \item $T_{1 \backslash 2} = Y^{1\to 1}$
    \item $T_{3 \backslash 2} = Y^{3\to 4}$
    \item $T_{4 \backslash 3} = Y^{4\to 4}$
\end{itemize}
\end{multicols}
The authorized pairs are graphically depicted by an edge in the access-pair graph, as shown in figure \ref{subfig:access-pair-example}.

Next, we would like to understand when the access pair graphs arising from an entanglement summoning task can be realized. 
From theorem \ref{thm:ESSwithunknownpartner}, we see that we need to understand the monogamy condition as applied to these graphs.
Unlike the bidirectional case, odd cycles are no longer sufficient to characterize the access-pair graph. 
For example, in figure \ref{fig:directed-sq}, there are no odd cycles, but monogamy is violated: There is an even length path $\{\{T_{1 \backslash 2}, T_2\}, \{T_2, T_{3\backslash 2}\}\}$, but $T_{1 \backslash 2} \cap T_{3 \backslash 2} = \emptyset$.  

\subsection{Sufficient conditions on the access pair graph}

In this section, we characterize when the access pair graphs arising from the protocol of the last section can be realized as entanglement sharing schemes. 
When they can be, this shows the associated entanglement summoning problem can be solved. 

Our characterization is handled in the following lemma, which gives conditions on the sharing graph $G_{C\rightarrow \mathcal{A}}$ for it to represent a realizable access structure. 
\begin{lemma}\label{lemma:sufficientaccesspair}
   $G_{C\rightarrow \mathcal{A}}$, an access pair graph constructed from a causal graph $G_C$, is realizable as an entanglement sharing scheme whenever the following conditions are satisfied.
    \begin{itemize}
        \item \textbf{(M1)} There is no odd-length path from $T_{i\backslash k}$ to $T_i$ for all $i,k\in\{1,...,n\}$,
        \item \textbf{(M2)} For each $i\in\{1,...,n\}$ where there exists $j_1,j_2\in\mathcal{O}_i$, $j_1\neq j_2$, it is the case that $\{T_{j_1}, T_{j_2}\}$ is \textbf{not} authorized.
         \item \textbf{(M3)} Whenever there is an even-length path between $T_{i_1}$ and $T_{i_2}$, $\{T_{j_1}, T_{j_2}\}$ must \textbf{not} be authorized $\forall j_1\in\mathcal{O}_{i_1}\backslash\{i_2\},j_2\in\mathcal{O}_{j_2}\backslash\{i_1\}$
        \item \textbf{(NOC)} There are no odd cycles.
    \end{itemize}
\end{lemma}
\begin{proof}
    From theorem \ref{thm:ESSwithunknownpartner}, we only need to check that the listed conditions imply monogamy is satisfied in $G_{C\rightarrow \mathcal{A}}$. 
    This means checking the monogamy condition for every possible even length path in the graph. 
    Recall that there are two types of vertices in $G_{C\rightarrow \mathcal{A}}$: body and wing vertices. 
    The two types of vertex lead to three types of path: 1) from a body to a body vertex, 2) from a body to a wing vertex, and 3) from a wing to a wing vertex. 
    Our proof involves treating each case separately and checking that the conditions listed in the lemma statement enforce monogamy. 

    \vspace{0.2cm}
    \noindent \textbf{Body-to-body path:} This follows from the same argument as in the bidirected case: consider an even length path from $T_i$ to $T_j$. Then by the NOC condition, $\{T_i,T_j\}$ are not neighbours in $G_{C\rightarrow \mathcal{A}}$, so they are neighbours in $G_C$. But then $T_i$ and $T_j$ will overlap as needed, and in particular share $Y^{j\to i}$ and $Y^{i\to j}$. 

    \vspace{0.2cm}
    \noindent \textbf{Body-to-wing path:} Consider an even length path that starts at $T_i$ and ends on $T_{j\backslash k}$. Note we do not need to consider the case where $i = k$, since $T_i$ and $T_{j\backslash i}$ are only ever connected by a path of length 1. This leaves two cases, $i=j$ and $i\neq j$.
    \begin{itemize}
        \item $i=j$: From equation \ref{eq:Tijdef} we have $T_{i\backslash k}=T_i\setminus \{Y^{k\to i}\}$ (that is, it has all the same subsystems as $T_i$ with one subsystem from $k$ removed), hence $T_{i\backslash k} \subset T_i$. Moreover, it is guaranteed that $T_{i\backslash k}\neq \emptyset$, thus $T_{i\backslash k}\cap T_i \neq \emptyset$.
        \item $i\neq j$: We are interested in even length paths from $T_i$ to $T_{j\backslash k}$, with $i,j,k$ all distinct. 
        To study this, consider the vertex $T_j$.
        We first claim $\{T_i,T_j\}$ is not authorized.
        To see why, suppose it \textit{is} authorized and hence share an edge in $G_{C\to\mathcal{A}}$. 
        We have assumed there is an even length path from $T_{j\backslash k}$ to $T_i$ so this gives an odd length path from $T_{j\setminus k}$ to $T_j$ (via $T_i$).
        But this contradicts M1, so $\{T_i,T_j\}$ is hence not authorized. 
        By remark \ref{remark:noedge}, this means $T_i\cap T_j \neq \emptyset$. 
        Furthermore, we know that $T_{j\backslash k}=T_j\setminus \{Y^{k\to j}\}$ and $Y^{k\to j}\not\subseteq T_i$, so it follows that $T_i\cap T_{j\backslash k} \neq \emptyset$, as needed.
    \end{itemize}

    \vspace{0.2cm}
    \noindent \textbf{Wing-to-wing path:}
    Considering vertices $T_{j_1\backslash i_1}$ and $T_{j_2\backslash i_2}$, we further separate this part of the argument into multiple cases:
    \begin{enumerate}
        \item $i_1 = i_2$, any $j_1, j_2$: This means $T_{j_1\backslash i_1}$ and $T_{j_2\backslash i_2}$ are wing vertices attached to the same body vertex. Letting $i_1=i_2=i$, they are both attached to $T_i$. From condition M2, we have that $\{T_{j_1},T_{j_2}\}$ is not authorized. 
        From remark \ref{remark:noedge} this means $T_{j_1}\cap T_{j_2} \neq \emptyset$. Then use that $T_{i \backslash j_1}=T_{j_1}\setminus Y^{i\to j_1}$ and $T_{i \backslash j_2}=T_{j_2}\setminus Y^{i\to j_2}$. Since $Y^{i \to{j_1}}$ is the system that can be sent from $i$ to $j_1$, and $Y^{i \to{j_2}}$ is the system sent from $i$ to $j_2$, neither of these can be in $T_{j_1}\cap T_{j_2}$, so it must be that $T_{j_1}\cap T_{j_2} = T_{i\backslash {j_1}}\cap T_{i\backslash{j_2}}$ and hence $T_{i\backslash {j_1}}\cap T_{i\backslash{j_2}}\neq \emptyset$ as needed. 
        \item $i_1 \neq i_2$ and $j_1 = j_2$: This means the wing vertices are attached to different body vertices. Let $j_1=j_2=:j$. Recall that $T_{j\backslash{i_1}}$ is comprised of the systems $D_j$ receives when it keeps its local systems, but it loses those that would come from $i_1$, while $T_{j\backslash{i_2}}$ are the systems that $D_j$ receives when it keeps its local systems but doesn't receive the system coming from $i_2$. Since there is always at least one subsystem starting at $j$ (since $\mathcal{O}_j$ is never empty), these sets must always overlap, as needed. 
        \item $i_1=j_2$ or $i_2 = j_1$: We show such an even length path can never exist. Suppose without loss of generality that $j_1 = i_2 =: l$. For sake of contradiction, assume that an even path from from $T_{{l}\backslash {i_1}}$ to $T_{{j_2}\backslash {l}}$ exists. Note that $T_{{j_2}\backslash {l}}$ is a wing vertex attached to $T_{l}$. Then, removing the last edge in the even path we assumed exists from $T_{{j_2}\backslash {l}}$ to $T_{{l}\backslash {i_1}}$, we have constructed an odd length path from $T_{{l}\backslash {i_1}}$ to $T_{l}$ which contradicts condition M1. Hence, such a path cannot exist.
        \item $i_1 \neq i_2\neq j_1\neq j_2$: Finally, the last case we have to consider is when each index is unique. Suppose there is an even-length path between $T_{j_1\backslash i_1}$ and $T_{j_2\backslash i_2}$. Now, consider a subpath with the last start and end edges $\{ T_{j_1\backslash i_1}, T_{i_1}\}$ and $\{ T_{i_2}, T_{j_2\backslash i_2}, \}$ removed. This is now an even length path between $T_{i_1}$ and $T_{i_2}$. By condition M3, the pair \{$T_{j_1}$, $T_{j_2}$\} is not authorized, hence $T_{j_1} \cap T_{j_2}\neq\emptyset$ by remark \ref{remark:noedge}. Then, using that $T_{j_1\backslash i_1}=T_{j_1}\setminus Y^{i_1\to j_1}$ and $T_{j_2\backslash i_2}=T_{j_2}\setminus Y^{i_2\to j_2}$, we have that $T_{j_1}\cap T_{j_2}=T_{j_1\backslash i_1}\cap T_{j_2\backslash i_2}\neq\emptyset $.
    \end{enumerate}
    Hence, every possible even length path in the access-pair graph satisfies monogamy, so the access pair graph is realizable as an entanglement sharing scheme. 
\end{proof}

\subsection{Sufficient conditions for the causal graph}

We claim that a causal graph $G_C$ can be implemented in a summoning problem whenever its associated sharing graph $G_{C\rightarrow \mathcal{A}}$ satisfies the conditions listed in lemma \ref{lemma:sufficientaccesspair}.
To understand these conditions better, we begin translating them from statements about the graph $G_{C\rightarrow \mathcal{A}}$ into statements about the original causal graph $G_C$.

To begin, we will first interpret the no odd cycles (NOC) condition in terms of the causal graph. 
This is similar to the discussion presented already in the bidirected case, where we found NOC implies the existence of a two-clique partition.
Recall that a clique is a subset of vertices in an undirected graph that all share an edge, and in the bidirected context we interpreted bidirected edges as undirected edges. In the context of directed graphs, we must use instead the language of quasi-cliques, which leads us to the following translation of NOC in terms of the causal graph.

\begin{lemma}\label{lemma:NOCstar}
    Consider a causal graph $G_C$ and the access pair graph $G_{C\rightarrow \mathcal{A}}$ constructed from it according to algorithm \ref{algorithm:accesspairgraphconstruction}. 
    Then the following two conditions are equivalent:
    \begin{itemize}
        \item (NOC): $G_{C\rightarrow \mathcal{A}}$ does not contain any odd length cycle.
        \item (NOC$\star$): $G_C$ admits a two-quasi-clique partitioning.
    \end{itemize}
\end{lemma}
\begin{proof}
    The arguments are similar to those in theorem \ref{thm:es-when-twoclique}. We first show that NOC implies NOC$\star$. 
    Suppose that $G_{C\to\mathcal{A}}$ does not contain any odd cycles. By theorem \ref{thm:bipartite-is-noc}, it follows that $G_{C\to\mathcal{A}}$ is bipartite, and admits a two-colouring. 
    Note that each wing vertex of $G_{C\to\mathcal{A}}$ is connected to a single body vertex, so they always have opposite colours to their attached body vertex in any two-colouring of $G_{C\to\mathcal{A}}$. 
    Hence, if the wing vertices were removed, the resulting graph would still admit a two-colouring and hence be bipartite. 
    Note that the body of $G_{C\to\mathcal{A}}$ is exactly the undirected complement of $G_C$. 
    It follows by theorem \ref{thm:two-quasi-clique-partition} that $G_C$ admits a two-quasi-clique partitioning. 

    In the other direction, suppose that $G_C$ admits a two-quasi-clique partitioning. By theorem \ref{thm:two-quasi-clique-partition}, this implies that $G^c_C$ is bipartite. $G^c_C$ is exactly the graph of the body of $G_{C\to\mathcal{A}}$, hence the body of $G_{C\to\mathcal{A}}$ is bipartite. Including the wing vertices, each of which are only connected to a single body vertex, the graph remains bipartite. Hence, all of $G_{C\to\mathcal{A}}$ is bipartite and thus, by theorem \ref{thm:bipartite-is-noc}, does not contain any odd cycles, as desired.
\end{proof}

The next simplest condition to express in the language of the causal graph is M2. 
In fact, this turns out to be the access pair language avatar of the `no two-out' condition already found in \cite{adlam2015quantum} and restated as lemma \ref{lemma:no-two-out}.
\begin{lemma}\label{lemma:M2star}
    Consider a causal graph $G_C$ and the access pair graph $G_{C\rightarrow \mathcal{A}}$ constructed from it according to algorithm \ref{algorithm:accesspairgraphconstruction}.
    Then the following two conditions are equivalent:
    \begin{itemize}
        \item (M2): For each $i\in\{1,...,n\}$ in $G_{C\rightarrow \mathcal{A}}$ where there exists $j,k\in\mathcal{O}_i$ it is the case that $\{T_{j}, T_{k}\}$ is \textbf{not} authorized.
        \item (M2$\star$): For any vertex $i$ in $G_C$ such that there exists vertices $j,k$ with $D_i\stackrel{!}{\to} D_j$, $D_i\stackrel{!}{\to} D_k$, we have that $D_j\sim D_k$.
    \end{itemize} 
\end{lemma}
\begin{proof}
    Beginning with M2, recall that $\mathcal{O}_i$ is the set of diamonds vertex $i$ points to in the causal graph, so we can identify $j,k\in \mathcal{O}_i$ as labelleling the diamonds $D_j$, $D_k$ that appear in the second condition. 
    Then from remark \ref{remark:noedge}, we know $\{T_j,T_k\}$ is not authorized if and only if $D_j \sim D_k$, so the two conditions are indeed equivalent. 
\end{proof}

Next, we move on to consider condition M1. 
We begin by observing that since M1 is a condition about odd length paths from a body to a wing vertex, we can rephrase it as a condition on even length paths between body vertices. 
In particular, we have the following claim. 
\begin{lemma} Consider a causal graph $G_C$ and the access pair graph $G_{C\rightarrow \mathcal{A}}$ constructed from it according to algorithm \ref{algorithm:accesspairgraphconstruction}.
    Then the following two conditions are equivalent:
    \item (M1): There is no odd-length path in $G_{C\rightarrow \mathcal{A}}$ from $T_{i\backslash k}$ to $T_i$ for all $i,k\in\{1,...,n\}$.
    \item (M1'): There is no even length path in $G_C^c$ from $D_k$ to $D_i$ for all diamonds $D_k\stackrel{!}{\to} D_i$.
\end{lemma}
\begin{proof}
    Note that $D_k\stackrel{!}{\to} D_i$ if and only if vertex $T_{i\backslash k}$ exists. 
    Next, note that $T_{i\backslash k}$ shares an edge only with $T_k$, so an odd length path between $T_i$ and $T_{i\backslash k}$ exists if only if an even length path between $T_i$ and $T_k$ exists. 
    Finally, edges between body vertices in $G_{C\rightarrow \mathcal{A}}$ are exactly the edges in $G_C^c$, so an even length path between $T_i$ and $T_k$ exists in $G_{C\rightarrow \mathcal{A}}$ if and only an even length path between $D_i$ and $D_k$ exists in $G_{C}^c$.
\end{proof}

This lemma does not yet translate M1 to the causal graph, but instead is phrased in terms of its complement. 
Our next task is to complete the translation. 
To do this, we use a result from graph theory. 
Specifically, because we already know our causal graphs satisfy NOC$\star$, we can use the following result about even length paths in co-bipartite graphs. 
\begin{lemma}\label{lemma:two-clique-for-M3}
    Suppose that $G$ is co-bipartite. There is an even length path between $v_1$ and $v_2$ in $G^c$ if and only if in every two-quasi-clique partition of $G$, $v_1$ and $v_2$ lie in the same clique.
\end{lemma}
\begin{proof}
    Consider any two-quasi-clique partition of $G=(V,E)$, say $V=Q_1\cup Q_2$. 
    Since $Q_1$ is a quasi-clique, all vertices in $Q_1$ are disconnected in $G^c$. The same holds true for $Q_2$ in $G^c$. 
    This means we can define a two colouring of $G^c$ by $G^c=Q_1\cup Q_2$. 
    Conversely, begin with any two colouring of $G^c$, $G^c=C_1\cup C_2$. 
    Then since no edge in $G^c$ connects any pair of vertices in $C_1$, it forms a quasi-clique in $G$. Similarly, $C_2$ is a clique in $G$. 
    Thus, the two-clique partitions of $G$ are exactly the two colourings of $G^c$. 
    
    Now suppose there is an even length path from $v_1$ to $v_2$ in $G^c$. This means they must have the same colour in any two colouring, and hence be in the same quasi-clique in $G$. 
    Conversely, if they are in the same quasi-clique in $G$ they must have the same colour in the corresponding two colouring of $G^c$, so they must be connected by an even length path.     
\end{proof}

This lets us complete the translation of M1 to the causal graph. 
\begin{lemma}\label{lemma:M1star} Consider a causal graph $G_C$ and the access pair graph $G_{C\rightarrow \mathcal{A}}$ constructed from it according to algorithm \ref{algorithm:accesspairgraphconstruction}.
    Then the following two conditions are equivalent:
\begin{itemize}
    \item (M1'): For diamonds $D_k\stackrel{!}{\to} D_i$, there is no even length path in $G_C^c$ from $D_k$ to $D_i$ .
    \item (M1$\star$): For all diamonds $D_k\stackrel{!}{\to} D_i$, there is at least one two-quasi-clique partition of $G_C$ such that $D_k$ and $D_i$ are in different cliques. 
\end{itemize}
\end{lemma}
\begin{proof}
    Follows from lemma \ref{lemma:two-clique-for-M3}. 
\end{proof}

Finally, we need to translate condition M3. 
\begin{lemma}\label{lemma:M3prime}
Consider a causal graph $G_C$ and the access pair graph $G_{C\rightarrow \mathcal{A}}$ constructed from it according to algorithm \ref{algorithm:accesspairgraphconstruction}.
    Then the following two conditions are equivalent:
\begin{itemize}
    \item (M3) Whenever there is an even-length path from $T_{i_1}$ to $T_{i_2}$, $\{T_{j_1}, T_{j_2}\}$ is not authorized for all $j_1\in \mathcal{O}_{i_1}\backslash\{i_2\}$ and $j_2\in\mathcal{O}_{i_2}\backslash\{i_1\}$
    \item (M3') There is no even length path in $G_C^c$ between $D_{i_1}$ and $D_{i_2}$ where $D_{i_1}\stackrel{!}{\to} D_{j_1}$ and $D_{i_2}\stackrel{!}{\to} D_{j_2}$ and $D_{j_1}\not\sim D_{j_2}$.
\end{itemize}
\end{lemma}
\begin{proof}
    Suppose that there is a path between $T_{i_1}$ and $T_{i_2}$, and that $\{T_{j_1}, T_{j_2}\}$ is not authorized in $G_{C\to\mathcal{A}}$. 
    Note that $j_1\in\mathcal{O}_{i_1}\backslash\{i_2\}$ and $j_2\in\mathcal{O}_{i_2}\backslash\{i_1\}$ if and only if $D_{i_1}\stackrel{!}{\to}D_{j_1}$ and $D_{i_2}\stackrel{!}{\to}D_{j_2}$.
    Moreover, $\{T_{j_1}, T_{j_2}\}$ is not authorized if and only if $D_{j_1}\sim D_{j_2}$. Note that $G^c_C$ is comprised of $\{T_i\}$ and the authorized pairings between them. Thus, all even-length paths between $T_{i_1}$ and $T_{i_2}$ having $\{T_{j_1}, T_{j_2}\}$ not authorized in $G_{C\to\mathcal{A}}$ is equivalent to every even-length path in $G^c_C$ between $D_{i_1}$ and $D_{i_2}$ requiring $D_{j_1}\sim D_{j_2}$ as required.
\end{proof}

Finally, lemma \ref{lemma:two-clique-for-M3} allows us to restate M3 in terms of the original graph.
\begin{lemma}\label{lemma:M3star}
Consider a causal graph $G_C$ and the access pair graph $G_{C\rightarrow \mathcal{A}}$ constructed from it according to algorithm \ref{algorithm:accesspairgraphconstruction}.
    Then the following two conditions are equivalent:
    \begin{itemize}
        \item (M3'): There is no even length path in $G_C^c$ between $D_{i_1}$ and $D_{i_2}$ where $D_{i_1}\stackrel{!}{\to} D_{j_1}$ and $D_{i_2}\stackrel{!}{\to} D_{j_2}$ and $D_{j_1}\not\sim D_{j_2}$.
        \item (M3$\star$):  Suppose that $D_{i_1} \stackrel{!}{\to} D_{j_1}$ and $D_{i_2} \stackrel{!}{\to} D_{j_2}$, and that $D_{i_1}$ and $D_{i_2}$ belong to the same quasi-clique for every two-quasi-clique partition of $G$. Then we must have $D_{j_1}\sim D_{j_2}$.
    \end{itemize}
\end{lemma}
\begin{proof}
    Follows from lemma \ref{lemma:two-clique-for-M3} and lemma \ref{lemma:M3prime}
\end{proof}
At this point, we can observe that condition M2 is actually a special case of condition M3 (corresponding to setting $i_1=i_2$). 

We can thus rephrase the previously found sufficient conditions of \ref{lemma:sufficientaccesspair} in terms of the original causal graph.
\begin{theorem}\label{thm:mainsufficiency}
    If a causal graph $G$ satisfies the following, then the entanglement summoning task is achievable.
    \begin{itemize}
        \item \textbf{(NOC$\star$)} $G$ admits a two-quasi-clique partition.
        \item \textbf{(M1$\star$)} If $D_{i_1} \stackrel{!}{\to} D_{j_1}$ then there is at least one two-clique partitioning of $G$ such that $D_{i_1}$ and $D_{j_1}$ belong to different quasi-cliques.
        \item \textbf{(M3$\star$)} If $D_{i_1}$ and $D_{i_2}$ belong to the same quasi-clique for every two-quasi-clique partition of $G$ and $D_{i_1}\stackrel{!}{\to} D_{j_1}$ and $D_{i_2}\stackrel{!}{\to} D_{j_2}$, then $D_{j_1}\sim D_{j_2}$.
    \end{itemize}
\end{theorem}
\begin{proof}
    This follows immediately from lemma \ref{lemma:sufficientaccesspair} and lemmas \ref{lemma:NOCstar}, \ref{lemma:M2star}, \ref{lemma:M1star}, and \ref{lemma:M3star}.
\end{proof}

The above conditions can be rephrased in terms of ``no one-way causal connections'' (which we denote $\not\leftarrow$) rather than ``exclusive one-way causal connections'' (i.e. $\stackrel{!}{\to}$). 
Doing so ends up reducing the number of conditions on the causal graph from four to two, as highlighted by the following corollary.\footnote{We thank Adrian Kent for pointing out to us how to reduce the number of conditions in this way.}

\begin{corollary}\label{coroll:mainsuff-rewritten}
    An entanglement summoning task is achievable whenever the following two conditions on its causal graph are satisfied:
    \begin{itemize}
        \item \textbf{(M1$\star\star$)} If $D_{i_1}\not\leftarrow D_{j_1}$, then there is at least one two-clique partitioning of $G$ such that $D_{i_1}$ and $D_{j_1}$ belong to different quasi-cliques.
        \item \textbf{(M3$\star\star$)} If $D_{i_1}$ and $D_{i_2}$ belong to the same quasi-clique for every two-quasi-clique partition of $G$ and $D_{i_1}\not\leftarrow D_{j_1}$ and $D_{i_2}\not\leftarrow D_{j_2}$, then $D_{j_1}\sim D_{j_2}$. 
    \end{itemize}
\end{corollary}

\noindent The above is proven in appendix 
\hyperref[sect:appendix-alt-conditions]{A}.

\subsection{Comments on necessity}\label{sec:necessity}

The above section found that the conditions (M1$\star$), (M2$\star$), (M3$\star$), and (NOC$\star$) on a causal graph $G$ are sufficient to show that an entanglement summoning task is achievable. This was done by mapping the problem of entanglement summoning to a problem of entanglement sharing using protocol \ref{protocol:mixeddirected}, and using the characterization of entanglement sharing by Khanian et al. \cite{Khanian2025} to determine which such graphs were allowed. 
In contrast to the bidirected case, however, this protocol may not be fully general and hence we have only established sufficiency of the conditions in theorem \ref{thm:mainsufficiency}.
We are able to prove necessity of (NOC$\star$) and (M2$\star$)\footnote{Recall that M2$\star$) is a weaker version of M2$\star$.}, but leave necessity of the remaining conditions open.

\begin{lemma}
    Suppose that an entanglement summoning task with causal graph $G$ is achievable. Then (NOC$\star$) is satisfied.
\end{lemma}
\begin{proof}
    Consider the task with causal graph $G'$, which is identical to $G$ except all one-way causal connections are replaced with bidirected connections. 
    This task is never harder than the original task.
    Then by theorem \ref{thm:bidirectedandNOC}, (NOC$\star$) must be satisfied. 
    That is, $G'$ admits a two-quasi-clique partition. 
    Since $G$ and $G'$ have the same edges (but perhaps not the same orientation of edges), $G$ will admit the same quasi-clique partitions as $G$, hence $G$ satisfies (NOC$\star$).
\end{proof}

\begin{lemma}
    Suppose that an entanglement summoning task with causal graph $G$ is achievable. Then (M2$\star$) is satisfied.
\end{lemma}
\begin{proof}
    The necessity of (M2$\star$) is proven by lemma 5 of \cite{dolev2021distributing}. In particular, it is shown that any entanglement summoning task that corresponds to a ``two-out'' causal graph (see figure \ref{fig:no-two-out}) is unachievable.
\end{proof}

We now turn our attention to the necessity of (M1$\star$) and (M3$\star$). 
The bidirected case admitted a complete characterization due to the fact that we were able to ignore the cases where neighbouring diamonds received calls. 
Since any two connected diamonds have a two-way causal connection, in the case where neighbouring diamonds received calls, they could see that their neighbour received a call, and they could complete the task using a dedicated Bell pair. By ignoring calls between neighbouring diamonds (i.e. assuming that all calls are made between causally disconnected diamonds), we were able to argue that protocol \ref{protocol:bidirected} was the most general protocol for completing the task. 
In the mixed case, we can no longer argue this way. 
We have not been able to find alternative arguments that either the given protocol which maps to an entanglement sharing protocol is fully general (we believe it is not), in which case necessity would follow, or (more likely) found alternative proofs that the remaining conditions are necessary. 

The simplest example of an uncharacterized problem is shown in figure \ref{fig:uncharacterized}. 
Note that making the single uni-directional edge there bidirected results in a feasible task, and removing the single uni-directed edge results in an impossible task (both claims can be seen using the characterization of bidirected examples).

\appendix

\section{Proof of corollary \ref{coroll:mainsuff-rewritten}}\label{sect:appendix-alt-conditions}

In this section, we prove that the conditions of corollary \ref{coroll:mainsuff-rewritten} are equivalent to the conditions in theorem \ref{thm:mainsufficiency}.

The proof that the double-starred conditions of corollary \ref{coroll:mainsuff-rewritten} imply the single starred conditions of theorem \ref{thm:mainsufficiency} is simple.
First, note that (M1$\star\star$) and (M3$\star\star$) immediately imply (M1$\star$) and (M3$\star$) respectively, just because $D_i\rightarrow D_j$ implies $D_i\not\leftarrow D_j$.  
Further, M1$\star\star$ implies NOC$\star$: if every pair of diamonds are connected bidirectionally, then automatically there is a two quasi-clique partition (since the entire graph is a clique). 
Otherwise, there is at least one pair of diamonds $D_1,D_2$ with $D_1 \not\leftarrow D_2$. 
But then M1$\star\star$ implies there is a quasi-clique.

To show the singly-starred conditions implies the double stars, we prove the following two lemmas.
\begin{lemma}\label{lemma:appendixA-1}
    (M1$\star$) and (NOC$\star$) imply (M1$\star\star$).
\end{lemma}
\begin{proof}
Suppose that $D_i \not\leftarrow D_j$. 
Then to prove (M1$\star\star$), we need to show there is a two-quasi-clique partition of $G$ such that $D_i$ and $D_j$ belong to different quasi-cliques. 
Since $D_i \not\leftarrow D_j$, we have that either $D_i \not\sim D_j$ or $D_i \stackrel{!}{\to} D_j$.

Suppose first that $D_i \not\sim D_j$. Then by NOC$\star$, there exists a two-quasi-clique partition with $D_i$ and $D_j$ in different cliques, as needed. 

Next, suppose that $D_i \stackrel{!}{\to} D_j$. 
Then M1$\star$ tells us there is a two-quasi-clique partitioning such that $D_i$ and $D_j$ belong to different cliques, as needed.
\end{proof}

\begin{lemma}\label{lemma:appendixA-2}
    If a causal graph $G$ satisfies (NOC$\star$), (M1$\star$) and (M3$\star$), then it also satisfies (M3$\star\star$).
\end{lemma}
\begin{proof} 
    Suppose that $D_{i_1}$ and $D_{i_2}$ belong to the same quasi-clique for every two-quasi-clique partition of $G$, and that  
    $D_{i_1}\not\leftarrow D_{j_1}$ and $D_{i_2}\not\leftarrow D_{j_2}$.
    We consider various cases. 

    First, suppose that $D_{i_1}\stackrel{!}{\to} D_{j_1}$ and $D_{i_2}\stackrel{!}{\to} D_{j_2}$. 
    Then M3$\star$ applies, so $D_{j_1}\sim D_{j_2}$ and we are done.

    Next, suppose that $D_{i_1}\not\sim D_{j_1}$ and $D_{i_2}\not\sim D_{j_2}$.
    Then $D_{j_1}$ and $D_{j_2}$ cannot be in the same quasi-clique as $D_{i_1}$ and $D_{i_2}$ respectively, so they must be in the other clique in any two-clique partition (which exists by NOC$\star$). 
    Thus, they are in the same quasi clique, and hence are connected, and we are done. 

    Next, suppose that one of the two pairings shares a causal connection; without loss of generality suppose that $D_{i_1}\stackrel{!}{\to} D_{j_1}$ and $D_{i_2}\not\sim D_{j_2}$.
    Then by M1$\star$, there is a two quasi-clique partition such that $D_{i_1}$ and $D_{j_1}$ are in different cliques, call these two cliques clique 1 and clique 2.
    Let clique 1 be the clique containing $D_{i_1}$, and clique 2 be the clique containing $D_{j_1}$.
    Note that $D_{i_2}$ is always in the same clique as $D_{i_1}$, so is also in clique 1. 
    Further, $D_{j_2}$ cannot be in clique 1 because $D_{i_2}\not\sim D_{j_2}$, so must be in clique 2. But then $D_{j_1}$ and $D_{j_2}$ are in the same clique, so they share a causal connection, as claimed.
\end{proof}

\section{Relationship to entanglement summoning tasks with oriented graphs}

In theorem \ref{thm:oriented-entanglement-summoning}, we recalled an earlier result that characterized all singly-directed entanglement summoning tasks.  
In this section, we show that the sufficient cases there are a subset of the cases shown to be sufficient by our theorem \ref{thm:mainsufficiency}. 
We proceed by showing that the conditions of theorem \ref{thm:oriented-entanglement-summoning} imply all of the conditions of theorem \ref{thm:mainsufficiency}.

\begin{lemma}
    \label{lem:orient-to-noc}
    Consider a causal graph $G_C$ that is an oriented graph and describes an entanglement summoning task. If every subgraph $G_{\mathcal{S}_j}$ of $G_C$ induced by $\mathcal{S}_{j} = \{ D_i\hspace{3pt} | \hspace{3pt}  D_i \not\to D_j\}$ is a tournament, then $G_C$ satisfies NOC$\star$. 
\end{lemma}

\begin{proof}
    If every pair of diamonds share an edge then any partition of the vertices is also a two-clique partition, so we are done. 
    Otherwise, we're missing at least one edge, say between $D_1$ and $D_2$. 
    We define the sets of vertices:
    \begin{align}
        K_1 &= \mathcal{S}_1 \setminus \mathcal{S}_2, \nonumber \\
        K_2 &= \mathcal{S}_2.
    \end{align}
    We claim that $K_1$, $K_2$ is a two quasi-clique partition. 
    To see why, first notice that $K_1\cap K_2=\emptyset$ by construction. 
    Further, $K_1$ and $K_2$ are both tournaments, since $\mathcal{S}_1$ and $\mathcal{S}_2$ are tournaments. 
    It remains to show that $K_1\cup K_2$ makes up the entire graph. 
    To show this, assume by way of contradiction that there exists a diamond, $D_k$, in neither $K_1$ nor $K_2$. 
    Since this means $D_k$ is in neither $S_1$ nor $S_2$, it means $D_k \stackrel{!}{\to} D_1$, $D_k \stackrel{!}{\to} D_2$. 
    But then $D_1,D_2$ must both be in $\mathcal{S}_k$, but since $D_1\not\sim D_2$ then $\mathcal{S}_k$ is not a tournament, which is a contradiction. 
\end{proof}

\begin{lemma}
    \label{lem:orient-to-m1s}
    Consider a causal graph $G_C$ that is an oriented graph and describes an entanglement summoning task. If every subgraph $G_{\mathcal{S}_j}$ of $G_C$ induced by $\mathcal{S}_{j} = \{ D_i\hspace{3pt} | \hspace{3pt}  D_i \not\to D_j\}$ is a tournament, then $G_C$ satisfies M1$\star$. 
\end{lemma}

\begin{proof}
    Suppose $D_{i_1}$ and $D_{j_1}$ are diamonds in $G_C$ that satisfy $D_{i_1} \stackrel{!}{\to} D_{j_1}$. Let $X, Y$ be a partitioning of $G_C$ such that (i) $D_{i_1} \in X$, (ii) $D_{j_1} \in Y$, (iii) $X$ induces a tournament in $G_C$, and (iv) $X$ is a maximal (in the number of diamonds) set with properties (i)-(iii).

    We claim that $Y$ must also induce a tournament in $G_C$. For the sake of contradiction, suppose otherwise. 
    Then there exist diamonds $D_s, D_t \in Y$ such that $D_s \not \sim D_t$. Dolev et al. \cite{dolev2021distributing} showed that theorem \ref{thm:oriented-entanglement-summoning} implies each diamond can only be non-adjacent to at most one other diamond. 
    Thus, $D_s$ and $D_t$ must be adjacent to all diamonds in $X$. At least one of $D_s \neq D_{j_1}$ or $D_t \neq D_{j_1}$ must hold, so we may add either $D_s$ or $D_t$ to $X$. Therefore, $X$ is not maximal in the number of diamonds, a contradiction.
\end{proof}

\begin{lemma}
    Consider a causal graph $G_C$ that is an oriented graph and describes an entanglement summoning task. If every subgraph $G_{\mathcal{S}_j}$ of $G_C$ induced by $\mathcal{S}_{j} = \{ D_i\hspace{3pt} | \hspace{3pt}  D_i \not\to D_j\}$ is a tournament, then $G_C$ satisfies M2$\star$. 
\end{lemma}

\begin{proof}
    Suppose for the sake of contradiction that M2$\star$ is not satisfied. That is, there exists $D_{i_1}\stackrel{!}{\to} D_{j_1}, D_{j_2}$ but $D_{j_1}\not\sim D_{j_2}$. But then it is the case that $D_{j_1}, D_{j_2}\in\mathcal{S}_{i_1}$, yet there is no edge between them, implying that $\mathcal{S}_{i_1}$ is not a tournament, which is a contradiction. 
\end{proof}

\begin{lemma}
    \label{lem:orient-to-m3s}
    Consider a causal graph $G_C$ that is an oriented graph and describes an entanglement summoning task. If every subgraph $G_{\mathcal{S}_j}$ of $G_C$ induced by $\mathcal{S}_{j} = \{ D_i\hspace{3pt} | \hspace{3pt}  D_i \not\to D_j\}$ is a tournament, then $G_C$ satisfies M3$\star$. 
\end{lemma}

\begin{proof}
    Since we already proved (M2$\star$), we can focus on the case where $i_1\neq i_2$.
    We show that the premise of M3$\star$ never holds in this setting.
    Consider an arbitrary pair of diamonds $D_{i_1}, D_{i_2}$ in $G_C$. Since $G_C$ is an oriented graph, we can assume (without loss of generality) that either (i) $D_{i_1} \not \sim D_{i_2}$ or (ii) $D_{i_1} \stackrel{!}{\to} D_{i_2}$.

    Suppose (i) holds. By lemma \ref{lem:orient-to-noc} we obtain NOC$\star$, which implies that $G_C$ admits a two-quasi-clique partition. If $D_{i_1}$ and $D_{i_2}$ appear in the same quasi-clique, then it must be that $D_{i_1} \sim D_{i_2}$, a contradiction.

    Suppose (ii) holds. By lemma \ref{lem:orient-to-m1s} we obtain M1$\star$, so there must exist a two-quasi-clique partition of $G_C$ such that $D_{i_1}$ and $D_{i_2}$ are in different quasi-cliques.

    Therefore, $D_{i_1}$ and $D_{i_2}$ cannot belong to the same quasi-clique for every two-quasi-clique partition of $G_C$, so M3$\star$ holds vacuously.
\end{proof}

Lemma \ref{lem:orient-to-m3s} demonstrates that M3$\star$ is an additional property that appears only with causal graphs that include bidirected edges. Overall, the preceding four lemmas demonstrate that theorem \ref{thm:oriented-entanglement-summoning} describes a strict subset of the entanglement summoning tasks provided by the conditions in theorem \ref{thm:mainsufficiency}.

\bibliographystyle{unsrtnat}
\bibliography{biblio}

\end{document}